\documentclass[11pt]{article}
\usepackage{fullpage}
\usepackage{rpmacros}
\RequirePackage[colorlinks=true]{hyperref}
\hypersetup{
  linkcolor=[rgb]{0,0,0.4},
  citecolor=[rgb]{0, 0.4, 0},
  urlcolor=[rgb]{0.6, 0, 0}
}
\usepackage{mathpazo}
\usepackage{bbm}
\usepackage{todonotes}
\usepackage{lipsum}
\usepackage{setspace}
\usepackage{mdframed}
\usepackage{tikz}
\usetikzlibrary{decorations.pathreplacing,backgrounds}

\usepackage[font=footnotesize]{caption}

\usepackage{amsthm}
\usepackage{thmtools,thm-restate}

\numberwithin{equation}{section}
\declaretheoremstyle[bodyfont=\it,qed=\qedsymbol]{noproofstyle}

\declaretheorem[name=Observation,numbered=no]{observation*}

\declaretheorem[numberlike=equation]{theorem}

\declaretheorem[name=Theorem,numbered=no]{theorem*}

\declaretheorem[numberlike=equation]{lemma}
\declaretheorem[name=Lemma,numbered=no]{lemma*}
\declaretheorem[numberlike=equation,style=noproofstyle,name=Lemma]{lemmawp}

\declaretheorem[numberlike=equation]{corollary}
\declaretheorem[name=Corollary,numbered=no]{corollary*}
\declaretheorem[numberlike=equation,style=noproofstyle,name=Corollary]{corollarywp}

\declaretheorem[name=Proposition,numbered=no]{proposition*}

\declaretheorem[name=Claim,numbered=no]{claim*}

\declaretheorem[name=Conjecture,numbered=no]{conjecture*}

\declaretheorem[name=Question,numbered=no]{question*}

\declaretheoremstyle[bodyfont=\it,qed=$\lozenge$]{defstyle} 

\declaretheorem[numberlike=equation,style=defstyle]{definition}
\declaretheorem[unnumbered,name=Definition,style=defstyle]{definition*}

\declaretheorem[unnumbered,name=Example,style=defstyle]{example*}

\declaretheorem[unnumbered,name=Notation=defstyle]{notation*}

\declaretheorem[unnumbered,name=Construction,style=defstyle]{construction*}

\declaretheorem[numberlike=equation,style=defstyle]{remark}
\declaretheorem[unnumbered,name=Remark,style=defstyle]{remark*}

\usepackage{nth}
\usepackage{intcalc}
\usepackage{etoolbox}
\usepackage{xstring}
\hypersetup{
}

\usepackage{ifpdf}
\ifpdf
\else
\usepackage[quadpoints=false]{hypdvips}
\fi

\newcommand{\ECCCurl}[2]{http://eccc.hpi-web.de/report/\ifnumcomp{#1}{>}{93}{19}{20}#1/#2/}

\newcommand{\shortECCC}[2]{\texttt{\href{http://eccc.hpi-web.de/report/\ifnumcomp{#1}{>}{93}{19}{20}#1/#2/}{eccc:TR#1-#2}}}

\newcommand{\parseECCC}[1]{
\StrSubstitute{#1}{TR}{}[\tmpstring]%
\IfSubStr{\tmpstring}{/}{ 
\StrBefore{\tmpstring}{/}[\ecccyear]%
\StrBehind{\tmpstring}{/}[\ecccreport]%
}{
\StrBefore{\tmpstring}{-}[\ecccyear]%
\StrBehind{\tmpstring}{-}[\ecccreport]%
}%
\shortECCC{\ecccyear}{\ecccreport}}

\newcommand{\occur}[3]{\mathrm{Occur}_{#1}^{#2}({#3})}
\newcommand{\var}[2]{\mathrm{Var}_{#1}({#2})}

\newcommand{\exi}{\vecx_{[i]}}

\newcommand{\cH}{\mathcal{H}}
\usepackage{algorithmicx}
\usepackage{algorithm} 
\usepackage{algpseudocode}

\algrenewcommand\algorithmicindent{1.0em}%

\def\rowsum{\operatorname{rowSum}} 
\def\colsum{\operatorname{colSum}}

\def\epsilon{\varepsilon} 
\let\eps\epsilon

\newcommand*\samethanks[1][\value{footnote}]{\footnotemark[#1]}

\date{}

\title{Identity Testing and Lower Bounds for Read-$k$ Oblivious Algebraic Branching Programs%
{\IfFileExists{./sha.tex}{\\\small SHA: \input{sha}}{}}}
\author{
Matthew Anderson\thanks{Department of Computer Science, Union College, Schenectady, New York, USA, E-mail: \texttt{andersm2@union.edu}}%
\and%
Michael A. Forbes\thanks{Department of Computer Science, Princeton University, USA, E-mail: \texttt{miforbes@csail.mit.edu}. Supported by the Princeton Center for Theoretical Computer Science.}%
\and%
Ramprasad Saptharishi\thanks{Department of Computer Science, Tel Aviv University, Tel Aviv, Israel, E-mails: \texttt{ramprasad@cmi.ac.in, shpilka@post.tau.ac.il, benleevolk@gmail.com}. The research leading to these results has received funding from the European Community's Seventh Framework Programme (FP7/2007-2013) under grant agreement number 257575.}%
\and%
Amir Shpilka\samethanks[3]
\and%
Ben Lee Volk\samethanks[3]
}
\begin{document}
\maketitle

\begin{abstract}
  Read-$k$ oblivious algebraic branching programs are a natural generalization of the well-studied model of read-once oblivious algebraic branching program (ROABPs).
In this work, we give an exponential lower bound of $\exp(n/k^{O(k)})$ on the width of any read-$k$ oblivious ABP computing some explicit multilinear polynomial $f$ that is computed by a polynomial size depth-$3$ circuit.
We also study the polynomial identity testing (PIT) problem for this model and obtain a white-box subexponential-time PIT algorithm.
The algorithm runs in time $2^{\tilde{O}(n^{1-1/2^{k-1}})}$ and needs white box access only to know the order in which the variables appear in the ABP.

\end{abstract}

\thispagestyle{empty}
\newpage
\pagenumbering{arabic}

\section{Introduction}
\label{sec:intro}

Algebraic complexity studies the complexity of syntactically computing polynomials using arithmetic operations.
The most natural model for computing polynomials is an {\em algebraic circuit}, which is a directed, acyclic graph whose leaves are labeled by either variables from $\set{x_1, \ldots, x_n}$ or elements from a field $\F$, and whose internal nodes use the arithmetic operations $+$ and $\times$.
Each node thus computes a polynomial in the natural way.
The associated complexity measures are the {\em size} (the number of wires) and the {\em depth} (the length of a longest path from an input node to the output node) of the circuit.
A circuit whose underlying graph is a tree is called a {\em formula}.

Another model of computation, whose power lies between that of circuits and formulas, is that of an {\em algebraic branching program} (ABP).
An ABP is a directed layered acyclic graph with a source node and a sink node, whose edges are labeled by polynomials.
An ABP computes a polynomial in the following way.
Every directed source-sink path computes the polynomial that is obtained from taking the product of all edge labels along the path.
The polynomial computed by the ABP is the sum over all paths of those polynomials.\footnote{This is analogous to boolean branching programs.
There each path computes the AND of edge labels and the output is the OR of all path-functions.}
Here, another relevant complexity measure is the {\em width} of the program, which is the maximal number of vertices in a layer (see \autoref{sec:models} for the exact definitions of the models that are considered in this work).

Two of the most important problems in algebraic complexity are (i) proving exponential lower bounds for arithmetic circuits (i.e., proving that any circuit computing some explicit polynomial $f$ must be of exponential size), and (ii) giving an efficient {\em deterministic} algorithm for the {\em polynomial identity testing} (PIT) problem.
The latter is the problem of given an arithmetic circuit, formula or ABP, computing a polynomial $f$, we have to decide whether $f$ is the identically zero polynomial.
PIT has a simple randomized algorithm that follows from the Schwartz-Zippel-DeMillo-Lipton lemma \cite{S80, Z79, DL78} that says that over a large enough field, a non-zero polynomial will evaluate to a non-zero value on {\em most} points.
Hence, in order to decide whether $f$ is zero it is enough to evaluate the circuit/formula/ABP on a random point (which can be done efficiently).

We further note that the randomized algorithm described above only needs to ability to evaluate $f$ at a given point.
Such algorithms are called {\em black-box} PIT algorithms.
It is readily seen that black-box algorithms are equivalent to producing a small {\em hitting set}, which is a set $\cH$ of evaluation points that has the property that $\cH$ contains a non-zero evaluation point for every non-zero $f$.
Algorithms that are given the computation graph as input are called {\em white-box} algorithms.
Naturally, white-box access is much less restrictive and one expects it will be easier to obtain better algorithms in this case.

Apart from being a very natural problem about arithmetic computation, PIT is one of the most general problems for which an efficient randomized algorithm is known, but no deterministic one.
Indeed, many other randomized algorithms --- e.g.\ parallel algorithms for finding matching in graphs \cite{KUW86, MVV87} or algorithms for polynomial factorization \cite{sv-icalp10, KSS15} --- reduce to PIT, in the sense that derandomization of PIT would derandomize those as well.

For more background on arithmetic circuits we refer the reader to the survey \cite{sy}.

\medskip

At first glance, the two problems described above seem rather different, as one is concerned with proving lower bounds and the other with providing efficient algorithms.
However, a series of works uncovered an intricate web of connections between the two, both in the white-box \cite{ki03, DSY09} and in the black-box \cite{HS80, a05} models.
That is, derandomizing PIT implies lower bounds for circuits (which gives a convincing explanation for why this problem is hard), and conversely, an explicit hard polynomial gives a recipe to ``fool'' small arithmetic circuits with respect to non-zeroness, in a very similar manner to the hardness-versus-randomness paradigm in boolean complexity.

In light of the hardness of proving lower bounds for general circuits, research has focused on trying to understand the effect that structural restrictions, like constant depth and multilinearity, have on the expressive power of the model.

One research direction that has attracted a lot of attention considers very shallow depth arithmetic circuits.
Following Valiant et al.\ \cite{vsbr83}, Agrawal and Vinay gave a reduction from general circuits to depth-$4$ circuits, that maps subexponential size to subexponential size \cite{av08}.
This reduction was later improved and extended in \cite{ koiran, Tav13, gkks13b}.
In a breakthrough work Gupta et al.\ \cite{gkks13} proved exponential lower bounds for depth-$4$ homogeneous formulas, which is the kind of circuit one gets from the reduction.
In the work that followed \cite{gkks13}, lower bounds for homogeneous depth-4 circuits were proved both for ``hard'' polynomials such as the permanent but also for easier polynomials such the determinant and the iterated matrix multiplication polynomial \cite{KSS13, FLMS13, KLSS, KS14, KS14a}.

In parallel, a lot of research effort was also focused on PIT for small-depth circuits with various restrictions such as bounded top fan-in or multilinearity \cite{DS07, KS07, ks09, SS12, kmsv13, sv11, osv15}.
Similar to the situation with lower bounds, a derandomization of PIT for depth-4 circuits (or, depth-3 in certain cases) implies a derandomization of the general case \cite{av08,gkks13b}.
As depth-$3$ multilinear formulas that have small top fan-in are a special case of sum of read-once arithmetic formulas (here, a read-once formula is an arithmetic formula in which each variable labels at most one node), Shpilka and Volkovich gave polynomial identity tests for this model \cite{SV15}.
Later, Anderson, van Melkebeek and Volkovich gave a PIT for multilinear read-$k$ formulas, which extend both models \cite{amv11}.

Another line of work focused on read-once oblivious ABPs (ROABPs, and we again refer to \autoref{sec:models} for the exact definition).
ROABPs were defined by Nisan \cite{nis91} in the context of proving lower bounds for non-commutative formulas.
While this model seems a bit restrictive, it was shown that derandomizing PIT for ROABPs implies derandomization of Noether's normalization lemma for certain important varieties \cite{Mulmuley12,FS13b}.
It is also not hard to show that ROABPs are strictly stronger than read-once arithmetic formulas.
Another motivation to study this model is that it is the algebraic analog of a boolean read-once branching program, which arises in the context of pseudorandomness for small-space computation \cite{Nisan92}.
Thus, one could hope for cross-fertilization of ideas between the models that could facilitate progress on both fronts.

Exponential lower bounds for ROABPs were known since their inception \cite{nis91}, and a white-box polynomial-time PIT algorithm was given by Raz and Shpilka \cite{RS05}.
In the black-box setting, hitting sets of quasipolynomial size were obtained in \cite{FS13, FSS14, agks15}, where the last two papers being applicable even if the order in which the variable are read is unknown.
This marks a striking difference between the algebraic model and the boolean model.
Indeed, in the boolean domain, pseudorandom generators for read-once branching programs in unknown order are much weaker, in terms of the seed length, than Nisan's generator \cite{Nisan92} which works only if the order is known.
Recently, Gurjar et al.\ obtained PIT algorithms for sum of ROABPs \cite{GKST15}.

In this work, we consider the natural next step, which are read-$k$ oblivious algebraic branching programs.
This model generalizes and extends both the models of ROABPs, of read-$k$ arithmetic formulas and of sum of ROABPs.
We are able to prove exponential lower bounds and to give subexponential-time PIT algorithms for this model.
A summary of our results appears in \autoref{sec:results}.

Prior to our work there were no results known for this model.
Some results were known for the more restricted model of a {\em sum} of $k$ ROABPs (e.g.\ \cite{GKST15}), and we give more details on those in \autoref{sec:work}.

\subsection{Computational Models}
\label{sec:models}
In this section we define the computational models we consider in this work.
We begin with the definition of {\em Algebraic Branching Programs} (ABPs).

\begin{definition}[Algebraic Branching Program, \cite{nis91}]
\label{def:abp}
An {\em Algebraic Branching Program (ABP)} is a directed acyclic graph with one vertex $s$ of in-degree zero (the {\em source}) and one vertex $t$ of out-degree zero (the {\em sink}).
The vertices of the graph are partitioned into layers labeled $0, 1, \ldots, L$.
Edges in the graph can only go from layer $\ell-1$ to layer $\ell$, for $\ell \in [L]$.
The source is the only vertex at layer $0$ and the sink is the only vertex at layer $L$.
Each edge is labeled with a polynomial in the input variables.
The {\em width} of an ABP is the maximum number of nodes in any layer, and the {\em size} of an ABP is the number of vertices in the ABP.
The {\em degree} of an ABP is defined to be the maximal degree of the polynomial edge labels.
	
Each path from $s$ to $t$ computes the polynomial which is the product of the labels of the path edges, and the ABP computes the sum, over all $s$ to $t$ paths, of such polynomials.
\end{definition}

The expressive power of ABPs lies between arithmetic formulas and arithmetic circuits.
Every formula of size $s$ can be simulated by an ABP of width $s$.
Similarly, an ABP of width $s$ and depth $d$ can be simulated by an arithmetic circuit of size $O(sd^2)$.

In this work we consider a restricted model of ABPs that we call read-$k$ oblivious ABPs.
In an oblivious ABP, in each layer all the labels are univariate polynomials in the same variable.
Furthermore, we also restrict each variable to appear in at most $k$ layers while still allowing them to label any number of the edges in those layers.

\begin{definition}[Read-$k$ Oblivious ABPs, \cite{FS13}]
\label{defn:read-k-abp}
  An algebraic branching program is said to be \emph{oblivious} if for every layer $\ell$, all the edge labels in that layer are univariate polynomials in a variable $x_{i_\ell}$.

Such a branching program is said to be a \emph{read-once} oblivious ABP (ROABP) if the $x_{i_\ell}$'s are distinct variables.
That is, each $x_i$ appears in the edge labels in at most one layer.

An oblivious ABP is said to be a \emph{read-$k$} if each variable $x_i$ appears in the edge labels of at most $k$ layers. 
\end{definition}

\begin{remark}
\label{rem:exactly-k}
For the rest of the discussion, it will be convenient to assume that in a read-$k$ oblivious ABP, every variable $x$ appears in {\em exactly} $k$ layers.
This assumption can be made without loss of generality, since if $x$ appears in $k' < k$ layers, we can add $k-k'$ ``identity'' layers to the program that vacuously read $x$.
This transformation does not increase the width of the program and increases the length by no more than $kn$.
\end{remark}

A special case of a read-$k$ oblivious ABP is one where the ABP makes ``multiple passes'' over the input.

\begin{definition}[$k$-pass ABPs]
\label{defn:k-pass-abp}
An oblivious ABP is said to be a \emph{$k$-pass ABP} if there exists a permutation $\pi$ on $n$ such that the ABP reads variables in the order
\[
\underbrace{x_{\pi(1)}, \ldots, x_{\pi(n)}, x_{\pi(1)}, \ldots, x_{\pi(n)}, \ldots, x_{\pi(1)}, \ldots, x_{\pi(n)}}_{\text{$k$ times}}.
\]
An oblivious ABP is said to be a \emph{$k$-pass varying-order ABP} if there are permutations $\pi_1, \cdots, \pi_k$ over $n$ symbols such that the ABP reads variables in the order
\[
x_{\pi_1(1)}, \ldots, x_{\pi_1(n)}, x_{\pi_2(1)}, \ldots, x_{\pi_2(n)}, \ldots, x_{\pi_k(1)}, \ldots, x_{\pi_k(n)}.\qedhere
\]
\end{definition}

\subsection{Our Results}
\label{sec:results}

We give various results about the class of read-$k$ oblivious ABPs, including lower bounds, PIT algorithms, and separations.

\paragraph{Lower Bounds:} We show an explicit polynomial $f$ such that any read-$k$ oblivious ABP computing $f$, for bounded $k$, must be of exponential width.

\begin{theorem}[proved in \autoref{sec:read-k ABP LBs}]
\label{thm:intro:lower-bound-k-abp}
There exists an explicit polynomial $f$, which is computed by a depth-3 polynomial-size multilinear circuit, such that any read-$k$ oblivious ABP computing $f$ must have width $\exp(n/k^{O(k)})$.
\end{theorem}

Prior to this work, there were no lower bounds for this model.

\paragraph{Identity Testing:} For the class of $k$-pass ABPs, we provide a black-box PIT algorithm that runs in quasipolynomial time.

\begin{theorem}[proved in \autoref{sec:pit-k-pass-abp}]
\label{thm:intro:pit-k-pass}
There exists a black-box PIT algorithm for the class of $n$-variate, degree-$d$, and width-$w$ $k$-pass oblivious ABPs that runs in time $(nw^{2k}d)^{O(\log n)}$.
\end{theorem}

For the more general class of read-$k$ oblivious ABPs, we provide a white-box PIT algorithm that runs in subexponential time.

\begin{theorem}[proved in \autoref{sec:pit-read-k-abp}]
\label{thm:intro:pit-k-abp}
There exists a white-box PIT algorithm for the class of $n$-variate, degree-$d$, and width-$w$ read-$k$ oblivious ABPs that runs in time $(nwd)^{\tilde{O}(n^{1-1/2^{k-1}})\cdot\exp(k^2)}$.
Furthermore, white-box access is only needed to know the order in which the variables are read.
That is, given this order, we construct an explicit hitting set of the above size for the class of read-$k$ oblivious ABPs that read their variables in that order.
\end{theorem}

\paragraph{Separations:}
Recently, Kayal, Nair and Saha \cite{KNS15} constructed a polynomial $f$ that can be computed by a sum of {\em two} ROABPs in different orders, each of constant width, such that any ROABP computing $f$ must be of width $2^{\Omega(n)}$.
Note that sum of two ROABPs is a special case of a $2$-pass varying-order ABP.

In order to exemplify the strength of the multiple-reads model, we show a polynomial that can be computed by a small 2-pass varying-order ABP, but cannot be computed by a small sum of ROABPs of small width.

\begin{theorem}[proved in \autoref{sec:separation}]
\label{thm:intro:separation}
There exists an explicit polynomial $f$ on $n^2$ variables that is computed by a 2-pass varying-order ABP of constant width, but any sum of $c$ ROABPs computing $f$ must be of width $\exp(\Omega(\sqrt{n}/2^c))$.
\end{theorem}

\subsection{Related Work}
\label{sec:work}

\paragraph{Algebraic Models}

As mentioned before, Nisan \cite{nis91} proved exponential lower bounds for ROABPs, and Raz and Shpilka \cite{RS05} gave a white-box polynomial-time PIT algorithm for this model.

Forbes and Shpilka \cite{FS13} were the first to consider the black-box version of this problem, and obtained a hitting set of size $(nwd)^{O(\log n)}$, for $n$-variate, degree-$d$ and width-$w$ ROABPs, if the order in which the variables are read is known in advance.
Forbes, Shpilka and Saptharishi \cite{FSS14} obtained a hitting set of size $(nwd)^{O(d \log(w) \log n)}$ for {\em un}known order ROABPs.
This was improved later by Agrawal et al.\ \cite{agks15} who obtained a hitting set of size $(nwd)^{O(\log n)}$ which matches the parameters of the known-order case.

For higher number of reads, much less was known.
Gurjar et al.
\cite{GKST15} considered the model of a {\em sum} of $c$ ROABPs, and obtained a white-box algorithm that runs in time $(ndw^{2^c})^{O(c)}$, and a black-box algorithm that runs in time $(ndw)^{O(c 2^c \log(ndw))}$, so that the running time is polynomial in the former case and quasipolynomial in the latter, when $c$ is constant.
A sum of $c$ ROABPs can be simulated by read-$c$ oblivious ABPs, and we show (in \autoref{sec:separation}) that read-$c$ oblivious ABPs are in fact strictly stronger.

Lower bounds against the model of sums of ROABPs were obtained in a recent work of Arvind and Raja \cite{AR15}, who showed that for every constant $\eps >0$, if the permanent is computed by a sum of $n^{1/2-\eps}$ ROABPs, then at least one of the ROABPs must be of width $2^{n^{\Omega(1)}}$.

We also mention an earlier work of Jansen et al.\ \cite{JQS10}, who also gave white-box and black-box tests for the weaker model of sum of constantly many read-once ABPs, where in their definition every variable is allowed to label only a single {\em edge} in the ABP.

Another model which is subsumed by oblivious read-$k$ ABPs is that of bounded-read formulas.
Shpilka and Volkovich \cite{SV15} constructed quasipolynomial-size hitting set for read-once formulas, and Anderson, van~Melkebeek and Volkovich \cite{amv11} extended this result to multilinear read-$k$ formulas and obtained a polynomial-time white-box algorithm and quasipolynomial-time black-box algorithm.
The natural simulation of read-$k$ formulas by ABPs produces an ABP in which every variable labels at most $k$ edges, and it can be seen that such programs can be converted to read-$k$ oblivious ABPs with only a polynomial overhead.

To conclude, earlier results apply only to restricted submodels of read-$k$ oblivious ABPs.

\paragraph{Boolean Models}

Let us now make a small detour and consider the boolean analogs for our models.
A (boolean) branching program is a directed acyclic graph with a source node $s$ and two sink nodes, $t_0$ and $t_1$.
Each internal {\em node} is labeled by a variable $x_i$ with two outgoing edges, labeled $0$ and $1$.
The program computes a boolean function on an input $(x_1, \ldots, x_n) \in \set{0,1}^n$ by following the corresponding path along the program.

A read-$k$-times boolean branching program is allowed to query every variable at most $k$ times along every path from the source a sink.
Note that this is more general than our definition of read-$k$ oblivious branching program.
Further distinction is made in the boolean case between {\em semantic} read-$k$ branching programs, in which this restriction is enforced only on paths that are consistent with some input, and between {\em syntactic} read-$k$ branching programs, in which this restriction applies for all paths (further note that in the read-once case, there is no distinction between the syntactic and the semantic model).

Exponential lower bounds for read-once branching program for explicit functions are known since the 1980's \cite{Zak84, BHST87, Wegener88}, even for functions that are computed by a polynomial size read-twice branching program.

Okolnishnikova \cite{Oko91}, and Borodin, Razborov and Smolensky \cite{BRS93} extended these results and obtained exponential lower bounds for syntactic read-$k$-times branching programs, by giving an explicit boolean function $f$ such that every syntactic read-$k$-times branching program for $f$ has size $\exp(n/2^{O(k)})$ (in fact, the lower bound in the second work also holds for the stronger class of non-deterministic branching programs).

A strong separation result was obtain by Thathachar \cite{Thathachar98}, who showed a {\em hierarchy theorem} for syntactic read-$k$-times boolean branching program, by giving, for every $k$, a boolean function $f$ which is computed by a linear-size syntactic read-$(k+1)$-times branching program such that every syntactic read-$k$-times branching program computing $f$ must have size $\exp(\Omega(n^{1/k}/2^{O(k)}))$.

The semantic model seemed more difficult, but nevertheless Ajtai \cite{Ajtai05} was able to prove an exponential lower bound for semantic read-$k$-times programs (when $k$ is constant), which was extended by Beame at al.
\cite{BSSV03} to randomized branching programs.

\medskip

PIT is the algebraic analog of constructing pseudorandom generators (PRGs) for boolean models.
A PRG for a class $\mathcal{C}$ of boolean circuits is an easily computable function $G : \set{0,1}^{\ell} \to \set{0,1}^n$, such that for any circuit $C \in \mathcal{C}$, the probability distributions $C(U_n)$ and $C(G(U_{\ell}))$ are $\epsilon$-close (where $U_m$ is the uniform distribution over $\set{0,1}^m)$.

Nisan \cite{Nisan92} constructed a PRG for polynomial size read-once oblivious branching programs with seed length $O(\log^2 n)$.
This was followed by a different construction with the same seed length by Impagliazzo, Nisan and Wigderson \cite{INW94}.
However, for the constructions to work it is crucial that the order in which the variables are read is known in advance.

Beyond that, and despite a large body of work devoted to this topic \cite{BDVY13, BPW11, BRRY14, De11, GMRTV12, IMZ12, KNP11, RSV13, Steinke12, SVW14}, all the results for the unknown order case or for read-$k$ oblivious branching programs have much larger seed length, unless further structural restrictions are put on the program (such as very small width, regularity, or being a permutation branching programs).
Specifically, we highlight that even for read-$2$ oblivious branching programs, the best result is by Impagliazzo, Meka and Zuckerman \cite{IMZ12} who gave a PRG with seed length $s^{1/2 + o(1)}$ for size $s$ branching program (note that the the dependence here is on $s$ rather than on $n$).
In particular, no non-trivial results are known for general polynomial size read-$2$ oblivious boolean branching program.

\subsection{Proof Technique}
\label{sec:technique}

Before delving into the details of our proof, it is perhaps instructive to think again about read-once branching programs.
The main exploitable weakness of these branching programs is that by the read-once property, their computation can be broken into two subcomputations over disjoint variables, that communicate with each other only through a small ``window'' of width $w$, the width of the branching program.
If $w$ is small it is natural to expect that upon reaching the middle layer, the branching program must ``forget'' most of the computation of the first half so that both subcomputations are ``almost independent'' in a way.
This property calls for a divide-and-conquer strategy, which was indeed, in very crude terms, the strategy that was applied both in the boolean model \cite{Nisan92} and in the algebraic model \cite{FS13, FSS14, agks15} (the details in each case, of course, are much more complicated than this simplistic description).

\subsubsection{Evaluation dimension and ROABPs}

Unfortunately, the above intuition breaks down when we allow a variable to be read multiple times, and this model requires a different strategy.
Our main starting point is the observation that, perhaps surprisingly, multiple ``passes'' over the input variables, {\em in the same order}, do not provide the program with much additional power.
That is, a $k$-pass ABP can be simulated by a ROABP, with a blow-up which is exponential in $k$ (hence, only a polynomial blow-up, if $k$ is constant).

This fact can be directly seen through analysis of the {\em evaluation dimension} measure.
For a polynomial $f(x_1, \ldots, x_n) \in \F[x_1, \ldots, x_n]$ and a subset of variables $S$, we denote by $\eval_S(f)$ the subspace of $\F[x_1, \ldots, x_n]$ that consists of all all the possible polynomials obtained from $f$ by fixing the variables in $S$ to arbitrary elements in $\F$.
The evaluation dimension of $f$ with respect to a partition $S, \overline{S}$, which is denoted $\evalDim_{S, \overline{S}}(f)$ is the dimension of $\eval_S(f)$.
Over large enough fields, this dimension equals the rank of the {\em partial derivative matrix} associated with this partition, as defined by Nisan \cite{nis91}.
In many contexts, however, it is easier to work with the evaluation dimension.
We refer to Chapter 4 of \cite{forbesphdthesis} for a detailed discussion on this equivalence, including formal proofs.

The importance of the evaluation dimension measure stems from the fact that $f$ can be computed by a width-$w$ ROABP in the order $x_1, x_2, \ldots, x_n$, if and only if $\evalDim_{\set{x_1, \ldots, x_i}, \set{x_{i+1}, \ldots, x_n}}(f) \le w$ for every $1 \le i \le n$.
Thus, this measure provides a precise characterization for the amount of resources needed to compute a polynomial in this model (see \autoref{thm:eval-dim-roabp}).

\subsubsection{Evaluation dimension and $k$-pass oblivious ABPs}

We are able to adapt the proof of the ``only if'' part of the above fact in order to show that if $f$ is computed by a $k$-pass oblivious ABP (that is, $f$ reads the $n$ variables $k$ times in the same order) then $\evalDim_{\set{x_1, \ldots, x_i}, \set{x_{i+1}, \ldots, x_n}}(f) \le w^{2k}$ for every $i \in [n]$.
That is, $k$ passes over the input in the same order cannot create many independent evaluations.
Then, using the ``if'' part of the equivalence, it follows that $f$ can also be computed using a ROABP of width $w^{2k}$ (see \autoref{lem:k-pass-to-roabp}).

This discussion immediately implies a hitting set of the class of $k$-pass oblivious ABPs of size $(ndw^{2k})^{O(\log n)}$ (\autoref{thm:intro:pit-k-pass}), as well as exponential lower bounds for this model, simply by applying the results for ROABPs.
It is still not clear, however, how to handle the general case, since even read-$2$ oblivious ABPs are exponentially stronger than ROABPs (recall that \cite{KNS15} give an exponential separation between a sum of two ROABPs and ROABPs, and we separate $2$-pass varying-order ABPs from sums of ROABPs).


\subsubsection{PIT for read-$k$ oblivious ABPs}

Let us focus, for the time being, on the simplest instance of the more general problem, by considering a $2$-pass varying-order ABP computing a non-zero polynomial $f$.
That is, an ABP of width $w$ that, without loss of generality, reads the variables in the order
$x_1, x_2, \ldots, x_n, x_{\pi(1)}, x_{\pi(2)},\ldots, x_{\pi(n)}$,
for some permutation $\pi$.
As we mentioned, we cannot possibly hope to simulate any such branching program by a small ROABP.
We do, however, find a large subset of the variables $S$, such that if we fix all the other variables arbitrarily (or, equivalently, think of $f$ as a polynomial in the variables of $S$ over the field of rational functions $\F(\overline{S})$), the resulting polynomial has a small ROABP.

By the well-known Erd\H{o}s--Szekers Theorem \cite{ES35}, any sequence of distinct integers of length $n$ contains either a monotonically increasing subsequence of length $\sqrt{n}$, or a monotonically decreasing subsequence of the same length.
Applied to the sequence $x_{\pi(1)}, x_{\pi(2)},\ldots, x_{\pi(n)}$ (with the natural order $x_1 < x_2 < \cdots < x_n$) we get a monotone subsequence of variables, which we might as well --- for the sake of this exposition --- assume to be monotonically increasing (the case of a decreasing sequence is, somewhat counter-intuitively, even simpler).
Let $S=\set{y_1, \ldots, y_{\sqrt{n}}}$ be the set of $\sqrt{n}$ elements that appear in this monotone subsequence.
Having fixed all the variables in $\overline{S}$, we are left, by the monotonicity property, with a branching program that reads the variables in the order $y_1, y_2, \ldots, y_{\sqrt{n}}, y_{1}, y_{2},\ldots, y_{\sqrt{n}}.$
Observe that this is exactly a 2-pass branching program!
Hence, the previous arguments apply here, and if $f$ is non-zero, we can efficiently find an assignment to the variables in $S$ from $\F$ that keeps the polynomial non-zero.
Having reached this point, we can ``resurrect'' the variables in $\overline{S}$, but note that we are left with only $n-\sqrt{n}$ variables.
These are again computed by a 2-pass varying-order ABP, so me may apply the same argument repeatedly.
After $O(\sqrt{n})$ iterations we are guaranteed to find an assignment to all the variables on which $f$ evaluates to a non-zero output.

At each stage, we construct a hitting set for width-$\poly(w)$ ROABPs, of size $(nwd)^{O(\log n)}$.
Since we take a cartesian product over $O(\sqrt{n})$ sets, the total size of the hitting set will eventually be $(nwd)^{\tilde{O}(\sqrt{n})}$, as promised by \autoref{thm:intro:pit-k-abp}.

Generalizing the argument above for $k$-pass varying-order ABPs is fairly straightforward, and is done using repeated applications of the Erd\H{o}s--Szekers Theorem to each of the $k$ sequences in order to obtain a subsequence of a subset of the variables $S$ which is monotone in every pass and has size only $n^{1/2^{k-1}}$, which accounts for most of the loss in the parameters.\footnote{This lower bound on the length of a subsequence which is monotone in every pass is the best possible.
This fact is attributed to de Bruijn (unpublished, see \cite{Kru53}), and the actual construction which shows that the lower bound is tight appears in \cite{AFK85}.}
The polynomial, restricted to variables in $S$, will be computed by a $k$-pass ABP.

In order to handle general read-$k$ oblivious ABPs, we need more ideas.
We observe that after repeatedly applying the Erd\H{o}s--Szekers Theorem to the subsequence of every ``read'', we do not get a $k$-pass ABP as before, but rather $k$ monotone sequences that are intertwined together.
We next show that by discarding more variables, but not too many, we get a structure that we call a ``$k$-regularly interleaving sequence''.
This is a technical notion which is presented in full details in \autoref{sec:pit-read-k-abp}, but the main point is that this definition allows us to argue that the obtained read-$k$ oblivious ABP has a (small) evaluation dimension and therefore it can be simulated by a not-too-large ROABP.
Obtaining this $k$-regularly interleaving property is the main technical difficulty of the proof.


\subsubsection{Lower bounds for read-k oblivious ABPs}

The arguments above that give PIT algorithms already give lower bounds for read-$k$ oblivious ABPs.
We have shown that if $f$ is computed by a 2-pass varying-order ABP of width $w$, then there exist a subset of $\sqrt{n}$ variables $S$ such that $f$ is computed by an ROABP of width $w^4$ over $\F(\overline{S})$.
This implies that if we pick $f$ so that every restriction to $\sqrt{n}$ variables has an exponential (in $\sqrt{n}$) lower bound for ROABPs, we would receive a subexponential lower bound for computing $f$ in a 2-pass varying-order ABP.
(These arguments, again, generalize to read-$k$ oblivious ABPs.)

In order to get an exponential lower bound (\autoref{thm:intro:lower-bound-k-abp}), we observe that we do not need to bound the evaluation dimension for {\em every} prefix (namely, to show that a subset of the variables is computed by a small ROABP), but only to show that the evaluation dimension is small for {\em some} prefix.
This is much easier to achieve since we do not need the order of the reads to be ``nicely-behaved'' with respect to {\em every} prefix, but just with respect to {\em a} prefix.

In other words, we invoke a simple averaging argument to show that if $f$ is computed by a width-$w$ read-$k$ oblivious ABP, then there exist sets of variables $S$ (of size at least $n/k^{O(k)}$) and $T$ (of size at most $n/100$), so that whenever we fix the variables in $T$ we get that $\evalDim_{S, \overline{S}} (g) \le w^{2k}$, where $g$ is any restriction of $f$ obtained by fixing the variables in $T$.
We then construct an explicit polynomial whose evaluation dimension with respect to every set remains large, even after arbitrarily fixing a small set of the variables (see \autoref{thm:read-k-eval-dim}).

\subsubsection{Separating 2-pass ABPs from sums of ROABPs}

In order to prove the separation with a 2-pass varying-order ABPs and sum of $c$ ROABPs (\autoref{thm:intro:separation}),
we use a structural result proved by Gurjar et al.\ \cite{GKST15} that gives a way to argue by induction on ROABPs. Given a polynomial $f$ which is computed by a sum $h_1 + h_2 + \cdots h_c$ of ROABPs of width $w$, we would like to find a related polynomial $f'$ that is computed by a sum of $c-1$ ROABPs of perhaps slightly larger width. Here, the evaluation dimension plays a role as well. The way to do this is to pick a non-trivial linear combination of $w+1$ partial evaluations of $f$ that make $h_1$ zero, which is possible since $h_1$ has a small evaluation dimension with respect to prefixes of variables corresponding to the order in which the variables are read in $h_1$. One can then show that, having eliminated $h_1$, each of the other summands can still be computed by a ROABP of width $w(w+1)$.

We provide a simple polynomial computed by a 2-pass varying-order ROABP whose partial evaluations are complex enough in the sense that they contain many linear independent evaluations and also a ``scaled-down'' version of the original polynomial as a projection.
It then follows by induction, using the above arguments, that this polynomial cannot be computed by a small sum of small ROABPs (see \autoref{lem:evalDim-rowsum-colsum}).

\subsection{Organization}

We start with some preliminaries and useful facts about the evaluation dimension in \autoref{sec:preliminaries} that almost all the results in this paper rely on.
In \autoref{sec:separation}, we present the separation between the class of $2$-pass varying order ABPs and sums of ROABPs.
Following that, in \autoref{sec:read-k ABP LBs}, we present an exponential lower bound for the class of general read-$k$ ABPs.
Then in \autoref{sec:pit-read-k-abp} we present the white-box PIT for read-$k$ ABPs.
Finally, we conclude with some open problems in \autoref{sec:conclusions}.

\section{Preliminaries}\label{sec:preliminaries}

\subsection{Notation}
For $n \in \N$, we denote by $[n]$ the set $\{1, 2, \ldots, n\}$. We commonly denote by $\vecx$ a set of $n$ indeterminates $\set{x_1, \ldots, x_n}$, where the number of indeterminates $n$ is understood from the context. As we often deal with prefixes of this set, we denote by $\exi$ the set $\set{x_1, \ldots, x_i}$, and more generally, for any $S \subseteq [n]$, $\vecx_S$ denotes the set $\set{x_i \mid i \in S}$.

For a polynomial $f \in \F[\vecx]$, a set $S \subseteq [n]$ and vector $\veca = (a_1, \ldots, a_{|S|}) \in \F^{|S|}$, we denote by $f|_{\vecx_S = \veca}$ the restriction of $f$ obtained by fixing the $j$-th element in $S$ to $a_j$.

For a subset $S \subseteq \vecx$ of variables, we denote its complement by $\overline{S}$. For disjoint subsets $S, T \subseteq [n]$ we denote by $S \sqcup T$ their disjoint union.

In our PIT algorithm, we need to combine hitting sets for smaller sets of variables. Hence, for a partition of $[n]$, $S_1\sqcup S_2\sqcup \cdots \sqcup S_m = [n]$, and sets $\cH_i \subseteq \F^{|S_i|}$, we denote by $\cH_1^{S_1}\times \cdots \times \cH_m^{S_m}$ the set of all vectors in $\F^n$ whose restriction to the $S_i$ coordinates is an element of $\cH_i$, that is
\[
\cH_1^{S_1}\times \cdots \times \cH_m^{S_m} = \{ v \in\F^n \mid \forall i\in [m], \; v|_{S_i} \in \cH_i  \}.
\]

We will also use the following theorem that gives a construction of a hitting set for ROABPs.

\begin{theorem}[Hitting Set for ROABPs, \cite{agks15}]
\label{thm:hitting-set-ROABP}
There exists a hitting set $\cH$ for the class of $n$-variate polynomials computed by width-$w$ individual-degree-$d$ ROABPs of size $(nwd)^{O(\log n)}$, in any variable order. $\cH$ can be constructed in time $\poly(|\cH|)$.
\end{theorem}

\subsection{ABPs and iterated matrix products}

The computation of an ABP corresponds to iterated multiplication of matrices of polynomials. In the case of oblivious branching programs, the ABP computes an iterated matrix product of univariate matrices. We record this fact as a lemma, and refer to \cite{forbesphdthesis} for a proof and a detailed discussion on this subject.

\begin{lemmawp}
\label{lem:obv-ABP}
Suppose $f$ is a polynomial computed by an oblivious ABP $A$ of width $w$ and length $\ell$, that reads the variables in some order $x_{i_1}, x_{i_2}, \ldots, x_{i_\ell}$. Then $f$ is the $(1,1)$ entry of a matrix of the form
\[
A_1(x_{i_1}) \cdot A_2(x_{i_2}) \cdots A_\ell(x_{i_\ell})
\]
where for every $j \in [\ell]$, $A_j \in \F[x_{i_j}]^{w \times w}$ is a $w \times w$ matrix in which each entry is a univariate polynomial in $x_{i_j}$.
\end{lemmawp}

\subsection{Evaluation dimension and ROABPs}

We now define a complexity measure for polynomials that we will use frequently when analyzing read-$k$ oblivious ABPs.

\begin{definition}[Evaluation dimension]
Let $f \in \F[x_1, \ldots, x_n]$ be a polynomial, and $S = \inbrace{x_{i_1},\ldots, x_{i_r}}$ be subset of variables.
We define $\eval_S(f)$ to be
\[
\eval_S(f) = \Span\setdef{f|_{\vecx_S=\veca}}{\veca \in \F^r} \subseteq \F[\overline{S}],
\]
which is the space of polynomials spanned by all partial evaluations of the $S$ variables in $f$. 

If $\vecx = S \sqcup T \sqcup R$ we define the \emph{evaluation dimension of $f$ with respect to $S \sqcup T$ over $\F(\vecx_R)$}, which shall be denoted by $\evalDim_{S,T;R}(f)$, as the dimension of the space $\eval_S(f)$ when taken over the field of rational functions $\F(\vecx_R)$. That is, we first ``move'' the variables $\vecx_R$ into the field and treat them as constants, and then consider the dimension of $\eval_S(f)$ over $\F(\vecx_R)$.

In the special case where $R=\emptyset$, we shall just use the notation $\evalDim_{S,T}(f)$.
\end{definition}

If $|\F| > \deg(f)$, then $\evalDim_{S,T}(f)$ is the rank of the \emph{partial derivative matrix with respect to $S$,$T$}, as defined by Nisan~\cite{nis91}. The rows of the partial derivative matrix are indexed by monomials $m_S$ in $S$ and its columns are indexed by monomials $m_T$ in $T$. The $(m_S,m_T)$ entry is the coefficient of $m_S m_T$ in the polynomial $f$. 
Although these two perspectives are equivalent, the formulation via evaluations is sometimes easier to work with. The evaluation dimension measure is useful when arguing about ROABPs since it characterizes the width needed to compute a polynomial $f$ using a ROABP. 

\begin{theorem}[\cite{nis91}, and see also \cite{forbesphdthesis}]
\label{thm:eval-dim-roabp}
  Let $f$ be a polynomial on $\vecx = \set{x_1,\ldots, x_n}$ and suppose for every $i \in [n]$ we have $\evalDim_{\exi,\overline{\exi}}(f) \leq w$.
Then, there is a ROABP of width $w$ in the order $x_1,\ldots, x_n$ that computes $f$.

Conversely, if $\evalDim_{\exi,\overline{\exi}}(f) = w$, then in any ROABP that computes $f$ in the order $x_1, x_2, \ldots, x_n$, the width of the $i$-th layer must be at least $w$.
\end{theorem}

Let us give an example of a polynomial which has large evaluation dimension with respect to a specific subset. This example will be helpful not only because it is simple to argue about, but also because all of our constructions of hard polynomials later on will ultimately be based on a reduction to this case.

\begin{lemma}
\label{lem:large-evaldim}
Let $f(\vecu, \vecv, \vecw)$ be a polynomial of the form
\[
f = \left( \prod_{i=1}^t (\ell_i(\vecv) + \ell'_i(\vecu)) \right) \cdot g(\vecu,\vecw),
\]
where:
\begin{enumerate}
\item For every $\veca \in \F^{|\vecu|}$, it holds that $g|_{\vecu=\veca}=g(\veca,\vecw) \not \equiv 0$.
\item $\set{\ell_i}_{i=1}^t$ is a set of linearly-independent linear functions , and so is $\set{\ell'_i}_{i=1}^t$.
\end{enumerate}
Then $\evalDim_{\vecu, \vecv \sqcup \vecw} (f) \ge 2^{t}$.
\end{lemma}

\begin{proof}
By applying a linear transformation to the variables (which cannot increase the dimension), if necessary, we may assume without loss of generality that
\[
f = \left( \prod_{i=1}^t (v_i + u_i) \right) \cdot g(\vecu,\vecw).
\]
In particular, for any $\veca=(a_1, \ldots, a_t) \in \set{0,1}^t$ we can evaluate $\vecu$ to $\veca$ so that
\[
 \left( \prod_{i=1}^t (v_i + a_i) \right) \cdot g|_{\vecu=\veca}(\vecw) \in \eval_{\vecu}(f).
\]
The $2^t$ polynomials $\prod_{i=1}^t (v_i + a_i)$ for $\veca \in \set{0,1}^t$ are linearly independent. Further, by the assumption on $g$, we also have that $g(\veca,\vecw)$ is non-zero. Hence these polynomials (in $\vecv$) remain linearly independent even when multiplied by the variable-disjoint polynomial $g(\veca, \vecw)$ and so $\evalDim_{\vecu, \vecv \sqcup \vecw} (f) \ge 2^{t}$.
\end{proof}

The following simple lemma is an illustration of using the evaluation dimension of a polynomial to obtain a small ROABP for that polynomial. 

\begin{lemma}\label{lem:k-pass-to-roabp}
Let $f \in \F[x_1, \ldots, x_n]$ be a polynomial computed by a $k$-pass ABP of width $w$, according to the order $\pi$.
Then $f$ can be computed by a width-$w^{2k}$ read-once ABP in the order $\pi$.
\end{lemma}
\begin{proof}
Let $A$ be the $k$-pass ABP computing $f$.  We may assume without loss of generality that the $k$ passes of $A$ read the variables in the order $x_1, \ldots, x_n$. Recall that for any $i \in [n]$, we denote $\exi = \{x_1, \ldots, x_i\}$.  By \autoref{thm:eval-dim-roabp}, it is enough to show that for any $i \in [n]$, 
\[
\evalDim_{\exi, \overline{\exi}} \le w^{2k}.
\] 
By the assumption on $f$ and by \autoref{lem:obv-ABP}, for every $i \in [n]$ and $j \in [k]$ there exists a matrix $M^{i,j} \in \F^{w \times w}$  such the entries of $M^{i,j}$ are univariate polynomials in $x_i$ and
\[
f = \left(  M^{1,1} (x_1) M^{2,1} (x_2) \cdots M^{n,1}(x_n) M^{1,2}(x_1) M^{2,2}(x_2) \cdots M^{n,k} (x_n) \right)_{1,1}.
\]
Fix $i \in [n]$, and consider any assignment of the form $\exi = \veca$ for $\veca=(a_1, \ldots, a_i) \in \F^i$. Having fixed $\exi$, we get that for some $k$ matrices $N_1(\veca), \ldots, N_k(\veca)$, that depend on $\veca$,
\begin{align}
\label{eq:k-pass}
f|_{\exi = \veca} &= \left( N_1(\veca) \cdot M^{i+1,1}(x_{i+1}) \cdots M^{n,1}(x_n) \cdot N_2(\veca) \cdot M^{i+1,2}(x_{i+1})  M^{n,2}(x_n)\right. \nonumber \\
&\cdots \left. N_k(\veca) \cdot  M^{i+1,k}(x_{i+1}) \cdots M^{n,k}(x_{n}) \right)_{1,1}.
\end{align}
It follows that any polynomial $g(x_{i+1}, \ldots, x_n) \in \eval_{\exi}(f)$ is completely determined by $N_1, \ldots, N_k$ which have $w^2$ entries each. More precisely, let $\set{B_1, \ldots, B_{w^2}}$ be a basis for $\F^{w \times w}$. For each $j \in [k]$, we can write $N_j(\veca) \in \F^{w \times w}$ in \eqref{eq:k-pass} as a linear combination of $\set{B_1, \ldots, B_{w^2}}$. Then, by expanding the matrix product in \eqref{eq:k-pass}, we see that every polynomial of the form $f|_{\exi = \veca}$ (and as a consequence, every polynomial in $\eval_{\exi}(f)$) is spanned by the $w^{2k}$ polynomials of the form
\[
\left( B_{\sigma_1} \cdot M^{i+1,1}(x_{i+1}) \cdots M^{n,1}(x_n) \cdot B_{\sigma_2} \cdot M^{i+1,2}(x_{i+1})  M^{n,2}(x_n) \cdots
 B_{\sigma_k} \cdot  M^{i+1,k}(x_{i+1}) \cdots M^{n,k}(x_{n}) \right)_{1,1}
\]
for $\sigma_1,\ldots,\sigma_k \in [w^2]$, which implies that $\evalDim_{\exi, \overline{\exi}}(f) \le w^{2k}$.

By \autoref{thm:eval-dim-roabp}, the claim follows.
\end{proof}

In fact, the proof of \autoref{lem:k-pass-to-roabp} permits a slight generalization of the lemma, by requiring weaker assumptions on the ABP, which is captured by the following definition.

\begin{definition}
\label{def:k-gap}
Let $A$ be an ABP that computes a polynomial $f \in \F[x_1, \ldots, x_n]$. We say that $A$ has the {\em $k$-gap property} with respect to $\{x_1, \ldots, x_i\}$, if there exist $k$ matrices $M_1, \ldots, M_k \in \F^{w \times w}[x_{i+1}, \ldots, x_n]$ such that for every $\veca \in \F^{i}$, there exists $k$ matrices $N_1(\veca), \ldots, N_{k}(\veca) \in \F^{w \times w}$ such that 
\begin{align}
\label{eq:k-gap}
f|_{\vecx_i = \veca} &= \left( \vphantom{M^{1,1}} N_1(\veca) \cdot M_1(x_{i+1},\ldots,x_n) \cdot N_2(\veca) \cdot M_2(x_{i+1},\ldots, x_n)\right. \nonumber\\
&\cdots \left. \vphantom{M^{1,1}}  N_k(\veca) \cdot  M_{k}(x_{i+1},\ldots,x_n) \right)_{1,1}.
\end{align}
$A$ is said to simply have the $k$-gap property if it has this property with respect to $\exi$, for every $i \in [n]$.
\end{definition}

\autoref{fig:k-gap} provides a pictorial explanation for the choice of this terminology.

\begin{figure}[h]
\begin{center}
\begin{tikzpicture}
\fill[red!20!white] (0,0) rectangle (2,1);
\fill[red!20!white] (4,0) rectangle (8,1);
\draw[ultra thick] (0,0) rectangle(12,1);

\foreach \i in {1,...,11} {
\draw (\i,0) -- (\i,1);
}

\node at (0.53,0.5) {\Large $x_1$};
\node at (1.53,0.5) {\Large $x_2$};
\node at (2.53,0.5) {\Large $x_3$};
\node at (3.53,0.5) {\Large $x_4$};

\node at (4.53,0.5) {\Large $x_1$};
\node at (5.53,0.5) {\Large $x_2$};
\node at (6.53,0.5) {\Large $x_1$};
\node at (7.53,0.5) {\Large $x_2$};

\node at (8.53,0.5) {\Large $x_3$};
\node at (9.53,0.5) {\Large $x_4$};
\node at (10.53,0.5) {\Large $x_3$};
\node at (11.53,0.5) {\Large $x_4$};

\end{tikzpicture}
\caption{An ABP that reads the variables in this (left-to-right) order is a read-$3$ ABP that has the $2$-gap property with respect to $\{x_1,x_2\}$.}
\label{fig:k-gap}
\end{center}
\end{figure}
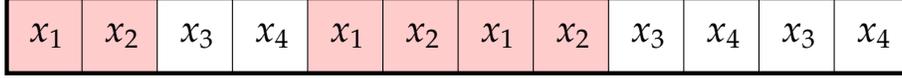

\noindent
Using the exact same arguments as in the proof of \autoref{lem:k-pass-to-roabp}, we obtain the following lemma.

\begin{lemmawp}
\label{lem:k-gap-to-roabp}
Let $f \in \F[x_1, \ldots, x_n]$ be a polynomial computed by an ABP of width $w$ with the $k$-gap property.
Then $f$ can be computed by a width-$w^{2k}$ read-once ABP.
\end{lemmawp}

\section{Separating $2$-pass ABPs from sums of ROABPs}
\label{sec:separation}

Recall that every {\em sum} of $c$ ROABPs can be realized by an oblivious read-$c$ ABP. In order to motivate our study of read-$k$ oblivious ABP, we begin by showing a polynomial that can be computed by a constant-width, $2$-pass varying-order ABP, and yet cannot be computed by a small sum of polynomial-size ROABPs. Thus, even a weak, but non-trivial, form of read-$k$ oblivious ABPs, for $k = 2$, is already stronger than sums of ROABPs.

Suppose $\vecx=\{x_{1,1},\ldots,x_{n,n}\}$ is a set of $n^2$ variables. It is useful to think of $\vecx$ as an $n \times n$ matrix $X$ such that $x_{i,j}$ appears in the $(i,j)$-th entry.
For every $m \in [n]$, define
\[
\rowsum_m = \sum_j x_{m,j} \quad \text{and} \quad \colsum_m = \sum_i x_{i,m}.
\]
Let
\begin{equation}
\label{eq:rowsum-colsum}
P_n(\vecx) = \inparen{\prod_{i=1}^n \rowsum_i} \cdot \inparen{\prod_{j=1}^n \colsum_j}.
\end{equation}

Observe that for all $i, j$, $\rowsum_i$ and $\colsum_j$ can be computed by width-2 ROABPs.  Moreover, both $\prod_{i=1}^n \rowsum_i$ and $\prod_{j=1}^n \colsum_j$ can be as well.  Indeed, their product $P_n$ is computed by a 2-pass varying-order ABP.

\begin{theorem}\label{thm:sep-2-pass-from-sum-ROABPs}
Let $P_n(x_{1,1},\ldots, x_{n,n})$ be the $n^2$-variate polynomial defined in \eqref{eq:rowsum-colsum}.
For every $c > 0$, any sum of $c$ ROABPs that computes it must have width $\exp(\sqrt{n}/2^{c})$. 
\end{theorem}

The proof exploits the structure of a sum of few ROABPs that Gurjar, Korwar, Saxena and Thierauf~\cite{GKST15} used for constructing hitting sets.
The following lemma is essentially present implicitly in their result. For completeness, we provide a proof.

\begin{lemma}[\cite{GKST15}]\label{lem:sum-of-ROABP-structure}
Let $f = h_1 + \cdots + h_c$ where each $h_i$ is computed by a width-$w$ ROABP in possibly different orders.
Then, for every $0 < t < n$, there exists a subset $S$ of $t$ variables such that for every set of $w+1$ partial assignments $\veca_1, \ldots, \veca_{w+1} \in \F^{t}$, there is some non-trivial linear combination of $\set{f|_{\vecx_S=\veca_i}}_{i=1}^t$ that is computable by a sum of $c-1$ ROABPs of width $w(w+1)$ in possibly different orders.
That is, there exists $\alpha_1, \ldots, \alpha_{w+1} \in \F$, not all zero, such that
\[
\sum_{i=1}^{w+1} \alpha_i \cdot f|_{\vecx_S = \veca_i} = f_1' + \cdots + f_{c-1}'
\]
where each $f_i'$ is a ROABP of width at most $w(w+1)$. 
\end{lemma}

\begin{proof}
Let $S$ be the first $t$ variables that are read in the ROABP that computes $h_1$. Since $h_1$ is computed by a width-$w$ ROABP, $\evalDim_{S, \overline{S}} (h_1) \le w$. Hence, every $w+1$ partial evaluations $\set{Q_1|_{\vecx_S=\veca_i}}_{i=1}^{w+1}$ are linearly dependent, that is, there exist  $\alpha_1, \ldots, \alpha_{w+1} \in \F$, not all zero, such that
\begin{equation}
\label{eq:linear-dependency}
\sum_{i=1}^{w+1} \alpha_i h_1|_{\vecx_S = \veca_i} = 0.
\end{equation}
Consider the polynomial $\sum_{i=1}^{w+1} \alpha_i \cdot f|_{\vecx_S = \veca_i}$. By the assumption on $f$,
\begin{align}
\label{eq:sum-of-c-1}
\sum_{i=1}^{w+1} \alpha_i \cdot f|_{\vecx_S = \veca_i} & = \sum_{i=1}^{w+1} \alpha_i \sum_{j=1}^c h_j |_{\vecx_S = \veca_i} \nonumber \\
&= \sum_{i=1}^{w+1} \alpha_i \cdot  h_1|_{\vecx_S = \veca_i} \;+\; \sum_{j=2}^c  \sum_{i=1}^{w+1} \alpha_i h_j |_{\vecx_S = \veca_i} \nonumber \\
&=  \sum_{j=2}^c  \sum_{i=1}^{w+1} \alpha_i h_j |_{\vecx_S = \veca_i},
\end{align}
where the last equality follows from \eqref{eq:linear-dependency}.

Hence, to prove the statement of the lemma it remains to be shown that for every $2 \le j \le c$, $\sum_{i=1}^{w+1} \alpha_i h_j |_{\vecx_S = \veca_i}$ is computed by a ROABP of width $w(w+1)$. Fix such $j$. Observe that since $h_j$ is computed by a ROABP of width $w$, for every $i \in [w+1]$ we have that $h_j |_{\vecx_S = \veca_i}$ is computed by a ROABP of width $w$ (by replacing the variables with the appropriate constants in the ABP that computes $h_j$), and furthermore all the ROABPs of the form $h_j |_{\vecx_S = \veca_i}$ for $i \in [w+1]$ are in the same order (inherited from the order of the ROABP computing $h_j$).

Therefore, we can connect those $(w+1)$ ROABPs in parallel to obtain a single ROABP, of width $w(w+1)$, computing  $\sum_{i=1}^{w+1} \alpha_i h_j |_{\vecx_S = \veca_i}$.
\end{proof}

The following lemma shows that the polynomial $P_n$ defined in \eqref{eq:rowsum-colsum} has many linearly independent partial evaluations. 
\begin{lemma}\label{lem:evalDim-rowsum-colsum}
  Let $S$ be any subset of $\vecx=\set{x_{1,1}, \ldots, x_{n,n}}$ of size $t < n$.
Then there exists $r \geq 2^{\sqrt{t}}$ partial evaluations $\veca_1, \cdots, \veca_r \in \set{0,1}^t \subseteq \F^t$ such that the polynomials $\set{P_n|_{\vecx_S = \veca_1}, \ldots, P_n|_{\vecx_S = \veca_r}}$ are linearly independent.

Furthermore, for any $g \in \Span\{P_n|_{\vecx_S = \veca_i} \mid i \in [r]\}$, there is a set $\vecy \subseteq \vecx \setminus S$ of $(n-t-1)^2$ variables, such that $P_{n-t-1}(\vecy)$ can be obtained as a projection of $g$: namely, for $\vecz = \vecx \setminus (\vecy \cup S)$ we can find $\veca \in \F^{n-|S|-|\vecy|}$ such that $ g|_{\vecz=\veca}=P_{n-t-1}(\vecy)$.
\end{lemma}
\begin{proof}
Recall that we think of the $n^2$ variables as an $n\times n$ matrix $X$.
By rearranging the rows and columns, assume that all variables in $S$ are present in the first $a$ rows and first $b$ columns.
Observe that $a,b \leq t$ and say $a \leq b$ so that we also have $b \geq \sqrt{t}$.
Also, any linear combination of evaluations of $S$ variables would always be divisible by $Q = \prod_{i>a} \rowsum_i \cdot \prod_{j > b} \colsum_j$ and hence we shall just work with 
\[
P' = \prod_{i=1}^a \rowsum_i \;\cdot\; \prod_{j=1}^b \colsum_j.
\]
In the $[a]\times[b]$ sub-matrix, set all variables not in $S$ to zero, and let $P''$ be the resulting polynomial.
Clearly, it suffices to establish linear independence of partial evaluations of $P''$. We label the remaining variables in the first $b$ columns by $\vecu$ if they belong to $S$ and by $\vecv$ otherwise. The variables which are not in the first $b$ columns are labeled by $\vecw$ (see \autoref{fig:matrix}).

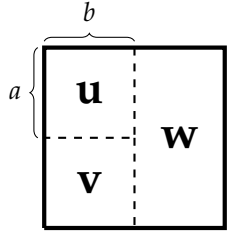
\begin{figure}[h]
\begin{center}
\begin{tikzpicture}[scale=0.8]

\draw[ultra thick] (0,0) -- (3,0) -- (3,3) -- (0,3) -- (0,0); 

\draw [decorate,decoration={brace,amplitude=3pt, mirror},xshift=-3pt]
(0,3) -- (0,1.5) node [black,midway,xshift=-0.3cm]
{\small $a$};

\draw [decorate,decoration={brace,amplitude=3pt},yshift=3pt]
(0,3) -- (1.5,3) node [black,midway,yshift=0.4cm]
{\small $b$};

\node at (0.75,2.25) {\LARGE $\vecu$};
\node at (0.75,0.75) {\LARGE $\vecv$};

\node at (2.25,1.5) {\LARGE $\vecw$};

\draw[thick, dashed] (1.5,3) -- (1.5,1.5) -- (0,1.5);
\draw[thick, dashed] (1.5, 1.5) -- (1.5, 0);
\end{tikzpicture}
\caption{Labeling of variables in the matrix. Variables in the top $[a] \times [b]$ submatrix which are not in $S$ are set to $0$.}
\label{fig:matrix}
\end{center}
\end{figure}

Then, we can write
\[
P'' = \prod_{i=1}^b (\ell_i(\vecv) + \ell'_i(\vecu)) \cdot g(\vecu, \vecw),
\]
so that we have the properties:
\begin{enumerate}
\item For $i \neq j$, $\ell_i$ and $\ell_j$ are supported on disjoint sets of variables, because they correspond to different column sums, and similarly for $\ell'_i$ and $\ell'_j$. In particular, each of the sets $\set{\ell_i}_{i=1}^t$ and $\set{\ell'_i}_{i=1}^t$ is linearly independent.
\item $g$ is the product the first $a$ row sums, so it is a product of variable-disjoint linear functions in $\vecu$ and $\vecw$, and in particular for any $\veca \in \F^{|\vecu|}$, $g(\veca,\vecw) \not\equiv0$. 
\end{enumerate}
By \autoref{lem:large-evaldim}, $\evalDim_{\vecu}(P'') \ge 2^b \ge 2^{\sqrt{t}}$. This implies that
\[
\evalDim_{S}(P') \ge \evalDim_{S}(P'') \ge \evalDim_{\vecu}(P'') \ge 2^{\sqrt{t}}.
\]
Hence, there exists $r \geq 2^{\sqrt{t}}$ evaluations $\veca_1, \ldots, \veca_r$ for the variables in $S$ for which the set of polynomials  $\inbrace{P'_n|_{\vecx_S = \veca_1}, \ldots, P'_n|_{\vecx_S = \veca_r}}$ are linearly independent. Hence, it also follows that the set of polynomials $\inbrace{P_n|_{\vecx_S = \veca_1}, \ldots, P_n|_{\vecx_S = \veca_r}}$ are linearly independent as well. This completes the first claim of the lemma. 

Now also observe that since $Q$ divides each $P_n|_{\vecx_S = \veca_i}$ it follows that any non-trivial linear combination of $\inbrace{P_n|_{\vecx_S = \veca_1}, \ldots, P_n|_{\vecx_S = \veca_r}}$ is a non-zero multiple of $Q$. Let us fix one such linear combination $g= h \cdot Q$ for a non-zero polynomial $h$. Note that $h$ depends on just the variables in the first $a$ rows and first $b$ columns. By the Schwartz-Zippel-DeMillo-Lipton lemma \cite{S80, Z79, DL78} there exists an assignment to the variables in the first $t$ rows and the first $t$ columns that keeps $h \cdot Q$ non-zero. If $\vecy' = \setdef{y_{ij}}{i\in [n-t], j\in [n-t]}$ is a relabeling of the variables in the last $(n-t)$ rows and columns, such an evaluation to the first $t$ rows and columns would result in a polynomial of the form
\[
Q' = \prod_{i=1}^{n-t} (\mathrm{rowSum}_i(\vecy')  + \alpha_i) \cdot \prod_{j=1}^{n-t} (\mathrm{colSum}_j(\vecy')  + \beta_j)
\]
for some field elements $\setdef{\alpha_i,\beta_j}{i,j\in [n-t]}$. By further setting $y_{i,n} = (-\alpha_i)$  and $y_{n,j} = (-\beta_j)$ for $i,j \in [n-t-1]$, and fixing $y_{n-t,n-t}$ to a value that preserves non-zeroness, we obtain the projection (up to a constant factor)
\[
Q'' = \prod_{i=1}^{n-t-1} \mathrm{rowSum}_i(\vecy) \cdot \prod_{j=1}^{n-t-1} \mathrm{colSum}_j(\vecy)
\]
where $\vecy = \vecy' \setminus \setdef{y_{ij}}{i=n \text{ or } j=n}$, which equals $P_{n-t-1}(\vecy)$. 
\end{proof}

With the above two lemmas, \autoref{thm:sep-2-pass-from-sum-ROABPs} is straightforward. 

\begin{proof}[Proof of \autoref{thm:sep-2-pass-from-sum-ROABPs}]
  The proof is a simple induction on $c$. We shall show that if $P_n$ is computable by a sum of $c$ ROABPs of width at most $w$, then
  \[
  n \spaced{\leq} \log^2(w+1) + \log^2\inparen{(w+1)^2} + \cdots \log^2 \inparen{(w+1)^{2^{c-1}}} + c. 
  \]
  Let us assume the hypothesis is true for $c-1$ and we now prove it for $c$. Suppose $P_n$ is computable by a sum of $c$ ROABPs of width $w$. Assume that $t = \log^2(w+1) < n$, for otherwise the lower bound follows immediately. By \autoref{lem:sum-of-ROABP-structure}, there is some set $S$ of $t$ variables such that for any $r = (w+1)$ partial evaluations $\veca_1,\ldots, \veca_r$ on $S$, some linear combination is computable by a sum of $c-1$ ROABPs of width $w(w+1)$.

On the other hand, \autoref{lem:evalDim-rowsum-colsum} states that we can find $r\geq 2^{\sqrt{t}} =w+1$ partial evaluations on $S$ that are linearly independent and any linear combination of them can be written has $P_{n-t-1}$ as a projection.

Therefore, if $P_n$ is computable by a sum of $c$ ROABPs of width $w$, then $P_{n-t-1}$ is computable by a sum of $(c-1)$ ROABPs of width at most $w(w+1) < (w+1)^2$, by taking the sum of $(c-1)$ ROABPs computing $\sum_{i=1}^{w+1} \alpha_i \cdot P_n |_{\vecx_S = \veca_i}$ and projecting further to obtain $P_{n-t-1}$.  But the inductive hypothesis then forces
\begin{align*}
n - \log^2(w+1) - 1 & \leq \inparen{\log^2(w+1)^2} + \log^2\inparen{(w+1)^{2^2}} + \cdots \log^2 \inparen{{\inparen{(w+1)^{2^{c-2}}}}^2} + (c-1)\\
\implies n  & \leq  \log^2(w+1) + \log^2(w+1)^2 + \log^2(w+1)^{2^2} + \cdots \log^2 (w+1)^{2^{c-1}} + c\\
 & \leq  4^c \cdot \log^2 (w+1) + c.
\end{align*}
Thus, $w \geq \exp(\Omega(\sqrt{n}/2^c))$ as claimed. 
\end{proof}

\section{Lower bounds for read-$k$ oblivious ABPs}\label{sec:read-k ABP LBs}

In this section we show an explicit polynomial that has a polynomial-size depth-$3$ multilinear circuit and yet cannot be computed efficiently by a read-$k$ oblivious ABP.

\subsection{An explicit polynomial with large evaluation dimension}

Raz and Yehudayoff~\cite{raz-yehudayoff} constructed an explicit multilinear polynomial $f(\vecx)$ with evaluation dimension as high as possible with respect to any partition $S, \overline{S}$.
Our requirements are slightly different, as we would need some ``robustness'' property, namely, we would want to argue that the evaluation dimension of the polynomial remains high even when we fix a small constant fraction (say, $n/10$) of the variables. Later, in \autoref{thm:poly-hard-for-read-k}, we show why this property implies hardness for read-$k$ oblivious ABPs.

Our construction is inspired by a recent similar construction of Kayal, Nair and Saha \cite{KNS15}.

Consider the complete bipartite graph $K_{n,n}$ with $n$ vertices on each side.
We shall label the left vertices as $x_1,\ldots, x_n$ and the right vertices as $y_1,\cdots, y_n$.
We can write $K_{n,n}$ as a union of $n$ edge-disjoint perfect matchings $M_1 \union \cdots \union M_n$, where for every $i \in [n]$, $M_j$ contains all edges of the form $(x_j, y_{j+i \bmod n})$ for $j \in [n]$.
Define the polynomial $Q_{n}$ as
\begin{equation}\label{eqn:Qn-defn}
Q_{n}(x_1,\ldots, x_n, y_1,\ldots, y_n,z_1,\ldots, z_n) = \sum_{i=1}^n z_i \prod_{(j,k) \in M_i} (x_j + y_k)
= \sum_{i=1}^n z_i \prod_{j=1}^n (x_j + y_{j + i \bmod n}).
\end{equation}

By its definition, it is clear that $Q_n$ is computed by a depth-$3$ polynomial-size circuit. We now show that even if we fix a small fraction of the variables in $\vecx \cup \vecy$, $Q_n$ retains a large evaluation dimension with respect to any partition of the variables we have not fixed.

\begin{lemma} \label{lem:high-evalDim-poly}
Let $S$, $T$ be two disjoint subsets of $\vecx \union \vecy$ such that $|S\sqcup T| \geq 0.9 \cdot 2n$.
Then, 
\[
\evalDim_{S,T}(Q_n) \spaced{\geq} \exp(\Omega(\min(|S|,|T|))).
\]
\end{lemma}
\begin{proof}
  Assume without loss of generality that $|S| \leq |T|$, and that $S_L := S \intersection \vecx$ satisfies $|S_L| \geq |S|/2$.
Since $(S \cup T) \cap \vecy \ge 0.8n$, $\abs{T \intersection \vecy} \geq (0.8n - |S|/2) \geq 0.3 n$.
Thus, there are $\Omega(n\cdot  |S|)$ edges between $S$ and $T$ in $K_{n,n}$. 
By averaging, some matching $M_i$ must include at least $\Omega(|S|)$ of these edges. Consider the polynomial $f_i = \prod_{(j,k)\in M_i} (x_j + y_k)$. As $\Omega(|S|)$ of  the edges in $M_i$ go between $S$ and $T$, we can write
\[
f_i = \prod_{m=1}^t (u_m + v_m) \cdot g(\vecw),
\]
where for every $m \in [t]$ we have that $u_m \in S$, $v_m \in T$, and $t=\Omega(|S|)$ (we have ``pushed'' to $g$ all the factors that correspond to edges in the matching which do not go between $S$ and $T$).

By \autoref{lem:large-evaldim}, $\evalDim_{S,T}(f_i) \geq 2^{\Omega(|S|)}$. Since $f_i$ is a projection of $Q_n$ (under the setting $z_i=1$ and $z_j=0$ for all $j \neq i$) it follows that $\evalDim_{S,T}(Q_n) \geq \evalDim_{S,T}(f_i) \geq \exp{(\Omega(|S|))}$.
\end{proof}

As an aside, we note that the difference between the above polynomial and  the one constructed by Kayal, Nair and Saha (\cite{KNS15}) is that they use a 3-regular bipartite expander instead of $K_{n,n}$ (which is important for their application). The degree of the graph corresponds to the number of matchings, and hence to the top fan-in of the depth-3 circuit computing the polynomial. In fact, in order to show hardness for read-$k$ oblivious ABPs it is also possible to use a good enough bipartite expander (with a constant degree $d$ that depends only on $k$) whose expansion property guarantees that a proof strategy along the lines of \autoref{lem:high-evalDim-poly} and \autoref{thm:poly-hard-for-read-k} would work. This would have allowed us to present a hard polynomial which is computed by a depth-$3$ circuit with bounded top fan-in, however, this very small gain would have come at the cost of increasing the complexity of the construction and making it depend on $k$. 

\subsection{Upper bound on evaluation dimension for read-$k$ oblivious ABPs}

In this section we show that if $f$ is computed by a read-$k$ oblivious ABP of width $w$, then we can fix a ``small'' subset of variables such that the remaining variables can be partitioned into
two carefully chosen ``large'' subsets, under which the evaluation dimension is at most $w^{2k}$. We then apply this result to the polynomial $Q_n$ (from \eqref{eqn:Qn-defn}) to show that if $Q_n$ is computed by a width-$w$ read-$k$ oblivious ABP, then $w \ge \exp(n/k^{O(k)})$. 

\begin{theorem}
\label{thm:read-k-eval-dim}
Let $f \in \F[x_1, \ldots, x_n]$ be a polynomial computed by a width-$w$ read-$k$ oblivious ABP. Then, there exist three disjoint subsets $U \sqcup V \sqcup W = [n]$, such that 
\begin{enumerate}
\item $|U|, |V| \ge n/k^{O(k)}$,
\item $|W| \le n/10$, and
\item $\evalDim_{U,V;W}(f) \le w^{2k}$.
\end{enumerate}
\end{theorem}

\begin{proof}
Consider an ABP $A$ that computes $f$.
Divide the $kn$ layers into $r$ equal-sized contiguous blocks of $kn/r$ layers (where $r$ shall be set shortly).
For each variable, consider the $k$ blocks that its $k$ reads fall in.
By a simple averaging, there must exist $k$ blocks $B_1,\ldots, B_k$ that contain all $k$ reads of a set $U$ of at least $n / \binom{r}{k}$ variables.
Let $W$ be the set of variables in $B_1 \cup B_2 \cup \cdots \cup B_k$ that are not in $U$, and $V$ be the set of all remaining variables.
As each block is of size $kn/r$, we have that $|W| \le k^2n/r$, which is at most $n/10$ if we set $r=10k^2$.  
Observe that $|V| \ge n - k^2n/r \ge 9n/10$.
Let us ignore the variables in $W$ by considering the ABP over the field $\F(\vecx_W)$.

We now claim that $\evalDim_{U,V;W}(f) \leq w^{2k}$. 
Having moved the variables in $W$ to the field, each of the $r$ blocks is either entirely contained in $U$ or entirely contained in $V$.
Therefore, since the reads comprise of at most $k$ alternating blocks of variables in $U$ and $V$, the resulting branching program has the $k$-gap property with respect to $U$. It follows immediately from \autoref{lem:k-gap-to-roabp} that $\evalDim_{U,V;W}(f)$ is at most $w^{2k}$. 
\end{proof}

We now show that $Q_n$ (defined in \eqref{eqn:Qn-defn}) is hard to compute for read-$k$ oblivious ABPs.

\begin{theorem}
\label{thm:poly-hard-for-read-k}
Let $A$ be a width-$w$, read-$k$ oblivious ABP computing the polynomial $Q_n$ (defined in \eqref{eqn:Qn-defn}). Then $w \ge \exp(n/k^{O(k)})$. 
\end{theorem}
\begin{proof}
First observe that we can eliminate the $\vecz$ variables by considering the ABP over the field $\F(\vecz)$
so that is now computes 
a polynomial in the variables $\vecx \cup \vecy$.

By \autoref{thm:read-k-eval-dim}, there exists a partition $U \sqcup V \sqcup W$ of $\vecx \cup \vecy$ with the prescribed sizes as in the statement of the theorem, such that $\evalDim_{U,V;W}(Q_n) \leq w^{2k}$.

Since $|W| \le 2n/10$, \autoref{lem:high-evalDim-poly} implies that 
$\evalDim_{U,V;W}(f) = \exp(\Omega(\min(|U|,|V|)))$.

Using the fact that $\min(|U|, |V|) \geq n/k^{O(k)}$, we get that $w^{2k} \ge \exp(n/k^{O(k)})$, which implies $w \ge \exp(n/k^{O(k)})$ as well.
\end{proof}

\section{Identity tests for read-$k$ oblivious ABPs}
\label{sec:pit-read-k-abp}

\subsection{Identity tests for $k$-pass ABPs}\label{sec:pit-k-pass-abp}

In this section we give PIT algorithms for the class of read-$k$ oblivious ABPs. First, observe that  \autoref{lem:k-pass-to-roabp} immediately implies a black-box algorithm for the subclass of $k$-pass ABPs, as those can be simulated efficiently by a ROABP.

\begin{corollary}\label{cor:k-pass-hittingset}
There is a hitting set of size $(ndw)^{O(k \log n)}$ for the class of $n$-variate $k$-pass ABPs of width $w$ and degree $d$. 
\end{corollary}
\begin{proof}
Follows directly from \autoref{lem:k-pass-to-roabp} and the $(ndw')^{O(\log n)}$-sized hitting set for width $w'$ read-once ABPs from \autoref{thm:hitting-set-ROABP}.
\end{proof}

We now turn to general read-$k$ oblivious ABPs.

\subsection{From read-$k$ to per-read-monotone and regularly-interleaving sequences}

In this section we show that given any read-$k$ oblivious ABP over $\vecx=\{x_1, \ldots, x_n\}$ computing a polynomial $f$, we can find a ``large'' subset of variables $\vecy \subseteq \vecx$ such that $f$ has a ``small'' ROABP when we think of $f$ as a polynomial in the $\vecy$ variables over the field $\F(\overline{\vecy})$. This process, in fact, involves only finding the correct subset $\vecy$ (without rewiring any part of the ABP).  Therefore, in order to avoid technical overhead it is useful to think in terms sequences over abstract sets of elements, which correspond to the order in which the ABP reads the variables, and not in terms of variables in branching programs.

Let $X$ be a set, and let $n=|X|$.  Let $S \in X^{m}$ be an sequence of elements from $X$.  We say $S$ is \emph{read-$k$} if each element $x \in X$ occurs $k$ times in $S$ (in this case we also have $m=nk$). 
As mentioned in \autoref{rem:exactly-k}, we will restrict ourselves to considering sequences that are read-$k$ for some $k$.
For $i \in [k]$, we denote by $S^{(i)}$ the subsequence of $S$ which consists of the $i$-th occurrences of elements in $X$. That is, $S^{(i)}$ is a permutation of the elements of $X$, according to the order in which they appear in $S$ for the $i$-th time. Similarly, for $i\neq j \in [k]$, we use the notation $S^{(i,j)}$ for the subsequence of $S$ which consists of the $i$-th and $j$-th occurrences of elements in $X$.

For $x \in X$ and $i \in [k]$ let $\occur{S}{i}{x}$ denote the index of the $i$-th occurrence of $x$ in $S$.  For an index $\ell \in [kn]$ let $\var{S}{\ell}$ denote the pair $(x_i,c)$ such that the $c$-th occurrence of $x_i$ appears at index $\ell$ in $S$.   For a subset $X' \subseteq X$, let $S|_{X'}$ denote the restriction of $S$ to the set $X'$ that is the result of dropping all elements of $X \setminus X'$ from $S$. Thus, $S|_{X'} \in X^{m'}$ for $m'=|X'|k$.

In order to save on excessive notation and multiple indexing, we will assume without loss of generality that $S^{(1)} = (x_1, \ldots, x_n)$, that is, that the variables in $\vecx$ are already labeled according to the order of their first occurrence. This can be ensured by renaming variables, if necessary. 

Next we define a special subclass of read-$k$ sequences which we work with throughout this section.

\begin{definition}
Let $S \in X^{nk}$ be a read-$k$ sequence. We say $S$ is {\em per-read-monotone} if for every $i \in [k]$, $S^{(i)}$ is monotone (that is, the variables all appear in either increasing or decreasing order).
\end{definition}

The following well-known theorem asserts that any long enough sequence contains a large monotone subsequence:

\begin{theorem}[Erd\H{o}s--Szekers Theorem, \cite{ES35, proofsfromthebook}]
\label{thm:erdos-szekers}
Let $S$ be a sequence of distinct integers of length at least $m^2+1$. Then, there exists a monotonically increasing subsequence of $S$ of length $m+1$, \emph{or} a monotonically decreasing subsequence of $S$ of length $m+1$.
\end{theorem}

As an immediate corollary of \autoref{thm:erdos-szekers}, we get the following lemma:
\begin{lemma}
\label{lem:2-per-read-mon}
Let $S$ be a read-$2$ sequence over $X=\{x_1, \ldots, x_n\}$. Then, there exists a subset $X' \subseteq X$ with $|X'| \ge \sqrt{n-1} +1 \ge \sqrt{n}$ such that the subsequence $S'=S|_{X'}$ is per-read-monotone.
\end{lemma}

\begin{proof}
First, observe that the subsequence of first occurrences is monotone by our definition of the order on $X$ according to the first occurrence in $S$ (and thus every subsequence of it is monotone as well).

Consider the subsequence $S^{(2)}$ of the $n$ second occurrences. By \autoref{thm:erdos-szekers}, there exists a monotonic subsequence of length at least $\sqrt{n}$. Let $X' \subseteq X$ be the set of elements that appear in this monotonic subsequence, and let $S'=S|_{X'}$. Then by the choice of $X'$, it follows that $S'$ is per-read-monotone.
\end{proof}

We can generalize \autoref{lem:2-per-read-mon} to read-$k$ sequences, at the cost of settling for a weaker lower bound of only $n^{1/2^{k-1}}$ on the length of the subsequence:

\begin{lemma}
\label{lem:k-per-read-mon}
Let $S$ be a read-$k$ sequence over $X=\{x_1, \ldots, x_n\}$. Then, there exists a subset $X' \subseteq X$ with $|X'| \ge n^{1/2^{k-1}} $ such that the subsequence $S'=S|_{X'}$ is per-read-monotone.
\end{lemma}
\begin{proof}
As in the proof of \autoref{lem:2-per-read-mon}, the set $X'$ can be constructed by repeatedly pruning $X$ using $k-1$ repeated applications of \autoref{thm:erdos-szekers}.

That is, we first apply \autoref{thm:erdos-szekers} on the subsequence $S^{(2)}$ of second occurrences and obtain a monotonic subsequence of length $n_2 := \sqrt{n}$. We discard all elements of $X$ which do not appear in this monotonic subsequence. We move on to the subsequence $S^{(3)}$ of third occurrences, and find a monotonic subsequence of length $n_3 := \sqrt{n_2}$, again discarding all elements that do not appear in this subsequence. After finding the $k$-th monotonic subsequence, we are left we a subset $X'$ of size $n^{1/2^{k-1}}$ that satisfies the conditions of the lemma.
\end{proof}

We now show how to prune per-read-monotone read-$2$ sequences even further, trading a constant fraction of their size for stronger structural properties. We begin by stating the property we look for.

\begin{definition}
\label{def:regularly-interleaving}
Let $S$ be a read-$2$ sequence over a set of elements $X$. We say $S$ is {\em 2-regularly-interleaving} if there exists a partition of $X$ to {\em blocks} $\{X_i\}_{i \in [t]}$ such that for every $i \in [t]$:
\begin{itemize}
\item For every $c \in \{1,2\}$, all the $c$-th occurrences of the block $X_i$ appear consecutively in $S$.
\item The interval containing the second occurrences of the block $X_i$ immediately follows the interval containing the first occurrences of $X_{i}$. 
\end{itemize}
A read-$k$ sequence $S$ is said to be {\em $k$-regularly-interleaving} if for any $i \neq j \in [k]$, the subsequence $S^{(i,j)}$ is 2-regularly-interleaving. That is, $S$ is $k$-regularly-interleaving if restricted to any two reads it is 2-regularly-interleaving.
\end{definition}
To get a better intuitive sense of the definition, the reader may consult \autoref{fig:2-regular-interleaving}.

\begin{figure}[h]
\begin{center}
\begin{tikzpicture}[framed]
\draw[ultra thick] (0,0) -- (8,0);
\draw[ultra thick,draw=red] (0,0) -- (2,0);
\draw[ultra thick,draw=red] (4,0) -- (6,0);

\draw [decorate,decoration={brace,amplitude=10pt},yshift=3pt]
(0,0) -- (2,0) node [black,midway,yshift=0.8cm]
{\footnotesize 1st occurrences of $X_1$};

\draw [decorate,decoration={brace,amplitude=10pt,mirror},yshift=-3pt]
(2,0) -- (4,0) node [black,midway,yshift=-0.8cm]
{\footnotesize 2nd occurrences of $X_1$};

\draw [decorate,decoration={brace,amplitude=10pt},yshift=3pt]
(4,0) -- (6,0) node [black,midway,yshift=0.8cm]
{\footnotesize 1st occurrences of $X_2$};

\draw [decorate,decoration={brace,amplitude=10pt,mirror},yshift=-3pt]
(6,0) -- (8,0) node [black,midway,yshift=-0.8cm]
{\footnotesize 2nd occurrences of $X_2$};

\node at (8.55,0) {\Large $\cdots$};

\draw[ultra thick] (9,0) -- (11,0);
\draw[ultra thick,draw=red] (11,0) -- (13,0);

\draw [decorate,decoration={brace,amplitude=10pt},yshift=3pt]
(9,0) -- (11,0) node [black,midway,yshift=0.8cm]
{\footnotesize 1st occurrences of $X_t$};

\draw [decorate,decoration={brace,amplitude=10pt,mirror},yshift=-3pt]
(11,0) -- (13,0) node [black,midway,yshift=-0.8cm]
{\footnotesize 2nd occurrences of $X_t$};

\end{tikzpicture}
\caption{A 2-regularly-interleaving sequence.}
\label{fig:2-regular-interleaving}
\end{center}
\end{figure}
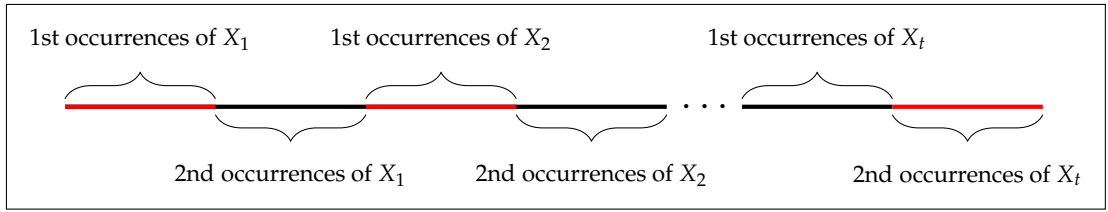

The following lemma is used to simplify some of the later arguments. It shows that in a read-$k$ per-read-monotone sequence, the monotonically increasing subsequences cannot intersect with monotonically decreasing subsequences.

\begin{lemma}
\label{lem:k-reg-interleave-concat}
Let $S$ be a read-$k$, per-read-monotone sequence over $X=\set{x_1, \ldots, x_n}$. Suppose $S^{(1)}$ is monotonically increasing. Then we can write $S$ as a concatenation
$S = (T_1 , T_2 , \ldots , T_t)$, such that:
\begin{enumerate}
\item for every $j \in [t]$, $T_j$ is a read-$k_j$ sequence for $k_j \le k$.
\item for every $i \in [k]$ there exists $j \in [t]$ so that $S^{(i)}$ is contained in $T_j$.
\item for every odd $j \in [t]$, all the subsequences $S^{(i)}$ that appear in $T_j$ are monotonically increasing, and for any even $j$, all are monotonically decreasing.
\item \label{item:border} for every $j \in [t-1]$, the last element that appears in $T_j$ equals the first element appearing in $T_{j+1}$, and this element can be either $x_n$ (if $T_j$ contains monotonically increasing subsequences and $T_{j+1}$ contains monotonically decreasing subsequences) or $x_1$ (in the opposite case). 
\end{enumerate}
In other words, we can partition $S$ into $t$ disjoint contiguous subsequences, such that every $S^{(i)}$ is completely contained in exactly one subsequence, and in every subsequence, either all reads are increasing or all reads are decreasing, with the pattern alternating.
\end{lemma}

\begin{proof}
The proof is by induction on $k$. For $k=1$, $S=S^{(1)}$ and this is a trivial statement.

For larger values of $k$, we would like to show first that no decreasing sequence can intersect an increasing one. Suppose without loss of generality that $S^{(2)}$ is decreasing (the other case, where the first sequence is decreasing and the second increasing, in handled analogously). Then,
\[
\occur{S}{2}{x_n} > \occur{S}{1}{x_n} > \occur{S}{1}{x_1},
\]
where the first inequality is obvious and the second follows from $S$ being per-read-monotone. Furthermore, since $S^{(2)}$ is decreasing, have that 
\[
\occur{S}{2}{x_n} < \occur{S}{2}{x_{n-1}} < \cdots < \occur{S}{2}{x_1},
\]
which implies that $S^{(2)}$ cannot intersect $S^{(1)}$, but rather it must begin after $S^{(1)}$ ends.

Let $\ell$ denote the first index in which a decreasing subsequence $S^{(j)}$ begins (if no such $\ell$ exists, the lemma is clearly satisfied by picking $T_1=S$). By the above argument, all the elements before the $\ell$-th index must belong to increasing subsequences which are read entirely. We can define $T_1$ to be the subsequence of $S$ from index $1$ up to index $\ell - 1$, and continue inductively on the subsequence $S'$ of $S$ from index $\ell$ to the end, which has $k' < k$ reads.

As for \autoref{item:border}, it follows from the fact that a sequence of monotonically increasing subsequences must end in $x_n$ (as to maintain monotonicity), and a sequence of monotonically decreasing subsequences must begin with $x_n$, for the same reason. The opposite case is handled analogously.
\end{proof}

The following lemma shows that given a 2-read per-read-monotone sequence, we can find a large subsequence which is also 2-regularly interleaving.

\begin{lemma}
\label{lem:2-per-read-regular}
Let $S$ be a read-$2$ per-read-monotone sequence over $X=\{x_1, \ldots, x_s\}$. Then there is a subset $X' \subseteq X$ with $|X'| \ge s/3$ such that the sequence $S' = S|_{X'}$ is per-read-monotone and 2-regularly-interleaving. 
\end{lemma}

\begin{proof}
We show how to erase the occurrences of (not too many) elements from $S$, such that the remaining sequence is 2-regularly-interleaving and maintains its per-read-monotonicity property.

First observe that if the subsequence $S^{(2)}$ of second occurrences is monotonically \emph{de}creasing, then, by \autoref{lem:k-reg-interleave-concat}, $S$ is already also 2-regularly-interleaving. In this case we have $S=(S^{(1)}, S^{(2)})$ and we can pick just one block, $X$, and satisfy the definition.

From now on we assume then that $S^{(2)}$ is monotonically \emph{in}creasing.
For every $z \in X$, denote by $d_z =  \occur{S}{2}{z} -  \occur{S}{1}{z}$ the distance between the first and the second occurrence of $z$ in $S$. Pick $x \in X$ such that $d_x$ is maximal and let $r:=d_x$. Among the $r$ occurrences between $\occur{S}{1}{x}$ and $\occur{S}{2}{x}$, there exist either $r/2$ first occurrences or $r/2$ second occurrences. Assume there are at least $r/2$ first occurrences (the other case is handled in an analogous way), and let $A$ be the set of variables (including $x$) whose first occurrence appears between the $\occur{S}{1}{x}$ and $\occur{S}{2}{x}$, so that $|A| \ge r/2$.

Since $S^{(2)}$ is monotonically increasing, for every $z \in A$ it holds that $\occur{S}{2}{z} > \occur{S}{2}{x}$. Let $y \in A$ be the element such that $\occur{S}{2}{y}$ is maximal. Observe that
\[
\occur{S}{2}{y} - \occur{S}{2}{x} \le \occur{S}{2}{y} - \occur{S}{1}{y} \le r,
\]
where the first inequality follows from the fact that $S$ is per-read-monotone and the second follows from the choice of $x$. Hence, it follows that
\begin{equation}
\label{eq:length-of-interval}
\occur{S}{2}{y} - \occur{S}{1}{x} \le 2r.
\end{equation}

We now erase from $S$ all the elements that appear in the interval $[ \occur{S}{1}{x} , \occur{S}{2}{y}]$ but do not appear in $A$.
Having done that, the subsequence in this interval satisfies the requirements of the lemma (one can relabel the elements if necessary in order to ensure contiguous indexing in this subsequence, and see also \autoref{fig:2-per-read-proof}). Furthermore, we have kept at least $|A| \ge r/2$ elements alive and erased, by \eqref{eq:length-of-interval}, at most $r$ elements.

\begin{figure}[h]
\begin{center}
\begin{tikzpicture}[framed]
\draw[ultra thick, draw=black!20] (0,0) -- (10,0);
\draw[ultra thick, draw=black] (2,0) -- (9,0);

\draw[thick] (2, -0.25) -- (2, 0.25);

\draw[thick] (6, -0.25) -- (6, 0.25);

\draw[draw=red] (3, -0.25) -- (3, 0.25);
\draw[draw=red] (4, -0.25) -- (4, 0.25);
\draw[draw=red] (4.5, -0.25) -- (4.5, 0.25);

\draw[draw=red] (7, -0.25) -- (7, 0.25);
\draw[draw=red] (7.5, -0.25) -- (7.5, 0.25);
\draw[draw=red] (9, -0.25) -- (9, 0.25);

\node at (2, -0.5) {$x$};
\node at (6, -0.5) {$x$};

\node at (4.5, -0.5) {$y$};
\node at (9, -0.5) {$y$};

\draw [decorate,decoration={brace,amplitude=10pt, mirror},yshift=-20pt]
(2,0) -- (4.5,0) node [black,midway,yshift=-0.8cm]
{\footnotesize 1st occurrences};

\draw [decorate,decoration={brace,amplitude=10pt, mirror},yshift=-20pt]
(6,0) -- (9,0) node [black,midway,yshift=-0.8cm]
{\footnotesize 2nd occurrences};

\end{tikzpicture}
\caption{Elements to be taken to $A$ are marked. All other elements in the black interval are discarded. The process then continues inductively on the gray subsequences.}
\label{fig:2-per-read-proof}
\end{center}
\end{figure}
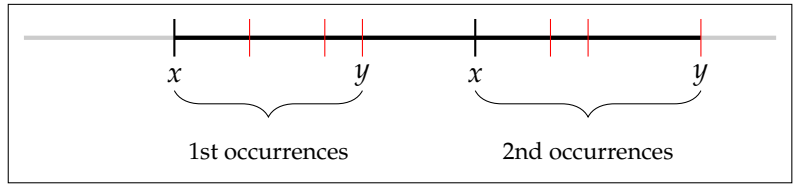

We continue recursively on the subsequences in both of the intervals $[1,\occur{S}{1}{x} - 1]$ and $[\occur{S}{2}{y} + 1, 2s]$. Observe that these intervals cannot share any element, as that would mean that the two occurrences of this element are of distance more than $r$ apart, which contradicts the choice of $x$. Hence, we may continue independently on both subintervals. By induction, the statement of the lemma follows.
\end{proof}

Viewed as an algorithmic process, the proof of \autoref{lem:2-per-read-regular} is a procedure that, given a per-read-monotone sequence $S$ over $X$, decides which elements of $X$ should be erased in order to be left with a 2-regular-interleaving sequence $S'=S|_{X'}$. It can also be noted that both properties of being per-read-monotone and being 2-regularly interleaving are downward-closed, in the sense that if we now take a subset $X'' \subseteq X'$ and look at $S''=S'|_{X''}$, it will maintain both properties. Hence, if we are given a read-$k$ per-read-monotone sequence $S$, by repeatedly applying the algorithmic process of \autoref{lem:2-per-read-regular} separately on each subsequence $S^{(i,j)}$ for $i\neq j \in [k]$ (maintaining a constant fraction of the elements on each application), we get the following corollary:

\begin{corollarywp}
\label{cor:k-per-read-regular}
Let $S$ be a read-$k$ per-read-monotone sequence over $X=\{x_1, \ldots, x_s\}$. Then there is a subset $X' \subseteq X$ with $|X'| \ge s/3^{k^2}$ such that the sequence $S' = S|_{X'}$ is per-read-monotone and $k$-regularly-interleaving.
\end{corollarywp}

\subsection{ROABPs for regularly interleaving sequences}
In this section we show that if a polynomial $f$ is computed by a small-width read-$k$ oblivious ABP $A$ such that the sequence $S$ of the reads in $A$ is per-read-monotone and $k$-regularly-interleaving, then $f$ can in fact also be computed by a small-width ROABP $A'$ (in the same order as $S^{(1)}$). We show this by proving that $A$ has the $k$-gap property with respect to that order, and then applying \autoref{lem:k-gap-to-roabp}.

\begin{lemma}
\label{lem:read-k-to-k-gap}
Let $f \in \F[x_1, \ldots, x_n]$ be computed by a read-$k$ oblivious ABP $A$ of width $w$, and let $S$ be the sequence of variables read by $A$. Suppose further that $S$ is per-read-monotone (with respect to the order $x_1 < x_2 < \cdots < x_n$) and $k$-regularly-interleaving. Then $A$ has the $k$-gap property.
\end{lemma} 

In the proof of \autoref{lem:read-k-to-k-gap} we will use the following lemma in order to bound the number of ``gaps'' one obtains for any prefix.

\begin{lemma}
\label{lem:k-gap-interface}
Let $S$ be a read-$k$, per-read-monotone, $k$-regularly-interleaving sequence over $\{x_1, \ldots, x_n\}$. Suppose that for every $i \in [k]$, $S^{(i)}$ is monotonically increasing. Let $\ell$ be an integer and suppose that $\var{S}{\ell}=(x_i, c)$\footnote{Recall that this notation means that the $\ell$-th element in $S$ is $x$, and this is its $c$-th occurrence.} and $\var{S}{\ell+1} = (x_j, d)$ with $j>i$. Then, it must be the case that $j=i+1$.
\end{lemma}

\begin{proof}
First observe that if $c=d$ the claim is true by monotonicity of $S^{(c)}$. If $c \neq d$, consider the subsequence $S^{(c,d)}$. Since $S$ is $k$-regularly interleaving, $S^{(c,d)}$ is 2-regularly interleaving, and in this sequence it also holds that $x_j$ immediately follows $x_i$. Furthermore, we have that $d < c$, as otherwise $\occur{S}{c}{x_j} < \occur{S}{c}{x_i}$, which contradicts the monotonicity of $S^{(c)}$, as we assumed that $j>i$. Hence, in $S^{(c,d)}$ $d$ plays the role of the {\em first} read, and $c$ plays the role if the {\em second} read. To avoid extraneous terminology let us assume for now that $d=1$ and $c=2$.

In a 2-regularly-interleaving sequence which is also per-read-monotone, the blocks in \autoref{def:regularly-interleaving} must be contiguous sequences of variables 
\[
\{x_1, x_2, \ldots, x_{i_1}\}, \{x_{i_1+1}, x_{i_1+2}, \ldots, x_{i_2}\}, ... \{x_{i_{t-1}+1}, \ldots, x_n\}.
\]

Suppose the blocks are indexed by $\{X_1, \ldots, X_t\}$, such that $x_i$ belongs to $X_{b_i}$ and $x_j$ belongs to $X_{b_j}$.

Recall that our assumption is that the first read of $x_j$ immediately follows the second read of $x_i$. Hence, after the block $X_{b_i}$ is read for the second time, we need to read the next block $X_{b_i+1}$ for the first time. This blocks contains $x_j$, so we get that $b_j=b_i+1$ and also that $x_j$ is the smallest element in that block, so $j=i+1$. 
\end{proof}

We are now ready to prove \autoref{lem:read-k-to-k-gap}.

\begin{proof}[Proof of \autoref{lem:read-k-to-k-gap}]
Let $S$ be the sequence of reads in $A$. Let $i \in [n]$. Consider any fixing $\exi = \veca$ for $\veca \in \F^i$. By plugging in the values to $A$, we can always write
\begin{align}
\label{eq:t-gap}
f|_{\exi = \veca} &= \left( \vphantom{M^{1,1}} N_1(\veca) \cdot M_1(x_{i+1},\ldots,x_n) \cdot N_2(\veca) \cdot M_2(x_{i+1},\ldots, x_n)\right. \nonumber\\
&\cdots \left. \vphantom{M^{1,1}}  N_t(\veca) \cdot  M_{t}(x_{i+1},\ldots,x_n) \right)_{1,1}.
\end{align}
for some integer $t$, where for each $\sigma \in [t]$, $N_\sigma$ is a product of univariate matrices of layers that read $\{x_1, \ldots, x_i\}$, and $M_\sigma$ is a product of univariate matrices of layers that read $\{x_{i+1}, \ldots, x_n\}$. We wish to show that $t$ can be at most $k$.

For each pair $(N_\sigma, M_\sigma)$ define their {\em interface} to be the pair $(\var{S}{\ell_\sigma}, \var{S}{\ell_\sigma+1})$, where $\ell_\sigma$ is the last index of a layer that participates in the product that defines $N_\sigma$, and thus $\ell_\sigma+1$ is the first index of a layer that participates in the product that defines $M_\sigma$ (see \autoref{fig:interface} for an illustration).

\begin{figure}[h]
\begin{center}
\begin{tikzpicture}

\node at (6,1.5) {$f|_{x_1=a_1, x_2=a_2} = \left( \vphantom{M^{1,2}} N_1(a_1, a_2) M_1(x_3, x_4) N_2 (a_1, a_2) M_2(x_3, x_4) \right)_{(1,1)}$};

\fill[red!20!white] (0,0) rectangle (2,1);
\fill[red!20!white] (4,0) rectangle (8,1);
\draw[ultra thick] (0,0) rectangle(12,1);

\foreach \i in {1,...,11} {
\draw (\i,0) -- (\i,1);
}

\node at (0.53,0.5) {\Large $x_1$};
\node at (1.53,0.5) {\Large $x_2$};
\node at (2.53,0.5) {\Large $x_3$};
\node at (3.53,0.5) {\Large $x_4$};

\node at (4.53,0.5) {\Large $x_1$};
\node at (5.53,0.5) {\Large $x_2$};
\node at (6.53,0.5) {\Large $x_1$};
\node at (7.53,0.5) {\Large $x_2$};

\node at (8.53,0.5) {\Large $x_3$};
\node at (9.53,0.5) {\Large $x_4$};
\node at (10.53,0.5) {\Large $x_3$};
\node at (11.53,0.5) {\Large $x_4$};

\node at (2, -0.3) {\small $\uparrow$};
\node at (2, -0.7) {\small $(N_1, M_1)$ interface};

\node at (8, -0.3) {\small $\uparrow$};
\node at (8, -0.7) {\small $(N_2, M_2)$ interface};

\end{tikzpicture}
\caption{The ABP reads the variable in the order that appears in the box. The locations of both interfaces are marked.}
\label{fig:interface}
\end{center}
\end{figure}
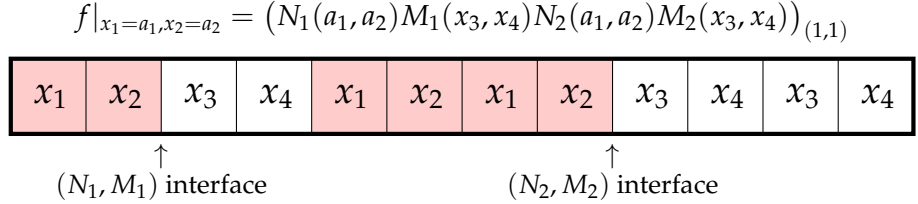

We first argue that we can assume, without loss of generality, that for every $j \in [k]$, $S^{(j)}$ is increasing. This is because \autoref{lem:k-reg-interleave-concat} implies that 
the sequence $S$ can be thought of as a concatenation of subsequences $(T_1, \ldots , T_t)$ that each have this property, and then handle each subsequence separately (one would of course need to prove an analog to \autoref{lem:k-gap-interface} and to some of the arguments we give further down the proof for the case where the subsequences in $T_m$ are decreasing. The analogs are rather straightforward and can be obtained by ``reversing'' the decreasing sequences. We leave the details to the reader). Furthermore, by \autoref{item:border} in \autoref{lem:k-reg-interleave-concat}, the ``border'' between $T_m$ and $T_{m+1}$ (that is, the last element of $T_m$ and the first in $T_{m+1}$) is always marked by either two occurrences of $x_1$ or of $x_n$.  Since $x_1$ is read at most $k$ times, it is immediate that any setting $x_1=a_1$ for $a_1 \in \F$ admits a representation of the form \eqref{eq:t-gap} for $t \le k$. The case where all variables $x_1, x_2, \ldots, x_n$ are fixed is trivial to verify. And for any other prefix $\exi$ these ``borders'' cannot increase the number of gaps in \eqref{eq:t-gap}.

From now on we assume then that for every $j \in [k]$, $S^{(j)}$ is increasing, and consider once again the representation \eqref{eq:t-gap}. We wish to show there cannot be too many interfaces. For every interface $\sigma \in [t]$, let us denote $\var{S}{\ell_\sigma} = (x_{\sigma_1}, c)$ and $\var{S}{\ell_\sigma+1} = (x_{\sigma_2}, d)$.

Since this is an interface, it must be the case that $\sigma_1 \le i < \sigma_2$. By \autoref{lem:k-gap-interface}, we must have $\sigma_2=\sigma_1+1$, and thus $\sigma_1 = i$ and $\sigma_2=i+1$. Hence, we can map each interface to a unique occurrence of $x_i$. This immediately implies there can be at most $k$ interfaces, and hence at most $k$ gaps.
\end{proof}

It now immediately follows that any read-$k$ oblivious ABP the reads the variables in a per-read-monotone and $k$-regularly-interleaving fashion can be simulated by a small ROABP. We record this fact in the following corollary.

\begin{corollary}
\label{cor:read-k-to-roabp}
Let $f \in \F[x_1, \ldots, x_n]$ be computed by a read-$k$ oblivious ABP $A$ of width $w$, and let $S$ be the sequence of variables read by $A$. Suppose further that $S$ is per-read-monotone (with respect to the order $x_1 < x_2 < \cdots < x_n$) and $k$-regularly-interleaving.
Then for any $i \in [n]$, $\evalDim_{\exi,\overline{\exi}} (f) \le w^{2k}$. In particular, $f$ is computed by a ROABP of width at most $w^{2k}$ in the variable order $x_1, x_2, \ldots, x_n$.
\end{corollary}

\begin{proof}
Immediate from \autoref{lem:read-k-to-k-gap} and  \autoref{lem:k-gap-to-roabp}.
\end{proof}

\subsection{Identity testing for read-$k$ oblivious ABPs}

In this section we give our white-box identity testing algorithm for read-$k$ oblivious ABPs. Before giving the proof, let us first give an overview of the algorithm for the slightly simpler read-$2$ case.

Given a read-$2$ oblivious ABP $A$ with read sequence $S$ which computes a polynomial $f \in \F[x_1, \ldots, x_n]$, \autoref{lem:2-per-read-regular} shows how to find a read-$2$ subsequence on a set $\vecy=\{y_1, \ldots, y_{\sqrt{n}}\}$ of roughly $\sqrt{n}$ variables, such that when we think of $f$ as a polynomial in the $\vecy$ variables over the field $\F(\overline{\vecy})$, it has a small ROABP. We can then use a hitting set for ROABPs in order to find an assignment (from $\F$) to the $\vecy$ variables that keeps the polynomial non-zero. Having done that, we are left with a non-zero polynomial over a smaller set of $n-\sqrt{n}$ variables, which is again computed by a read-$2$ oblivious ABP, so we may repeat this process. After at most $O(\sqrt{n})$ iterations we find an assignment for all the variables that keeps the polynomial non-zero. We note that a very similar ``hybrid argument'' that uses a hitting set for ROABPs appears both in \cite{agks15} and \cite{osv15}.

The argument for read-$k$ is identical, apart from the loss in the parameters incurred by \autoref{cor:k-per-read-regular}.

\begin{theorem}
\label{thm:read-k-pit}
There is a white-box polynomial identity test for read-$k$ oblivious ABPs of width $w$ and degree $d$ on $n$ variables that runs in time
$\poly(n,w,d)^{n^{1-1/2^{k-1}} \exp(k^2) \polylog(n)}$.
Furthermore, given only the order in which the variables are read, we can construct a hitting set for such ABPs that read their variables in this order, of size $\poly(n,w,d)^{n^{1-1/2^{k-1}} \exp(k^2) \polylog(n)}$.
\end{theorem}

We begin with a technical lemma which we use to bound the number of iterations of the above process. 

\begin{lemma}
\label{lem:iterations-inequality}
Let $p$ be a real number such that $0<p<1$ and $r$ be a positive integer. For any $n \in \N$,
\[
n^{1-p} - (n-n^p/r)^{1-p} \ge (1-p)/r.
\]
\end{lemma}
We defer the elementary proof of \autoref{lem:iterations-inequality} to \autoref{app:iterations-inequality}.

\medskip

Our PIT algorithm is presented in \autoref{alg:k-read-pit}.

\begin{algorithm}
  \caption{: PIT for read-$k$ oblivious ABPs}
  \label{alg:k-read-pit}
\begin{algorithmic}[1]
  \Require{a read-$k$ oblivious ABP $A$ computing a polynomial $f \in \F[x_1, \ldots, x_n]$.}
  \State{$\vecx=\{x_1, \ldots, x_n\}$, $i=1$}
  \While{$\vecx \neq \emptyset$}
  \State \parbox[t] {\dimexpr\linewidth-\algorithmicindent}
  {Pick a subset $\vecy_i \subseteq \vecx$ of size at least $|\vecx|^{1-1/2^{k-1}}/3^{k^2}$, such that the subsequence that reads only the $\vecy_i$ variables is per-read-monotone and $k$-regularly-interleaving (such a subset exists by \autoref{lem:k-per-read-mon} and \autoref{cor:k-per-read-regular}).\strut}
  \State \parbox[t] {\dimexpr\linewidth-\algorithmicindent}   
 {Construct a set $\cH_i \subseteq \F^{|\vecy_i|}$ of size $(nw^{2k}d)^{O(\log n)}$ that hits ROABPs of width $w^{2k}$ in the $\vecy_i$ variables, using \autoref{thm:hitting-set-ROABP}.\strut}
 \State{$\vecx \leftarrow \vecx \setminus \vecy_i$, $i\leftarrow i+1$}
  \EndWhile
  \State{{\bf return} the set $\cH = \cH_1^{\vecy_1} \times \cdots \times \cH_t^{\vecy_t} $ (where $t$ is the number of iterations of the loop).}
\end{algorithmic}
\end{algorithm}

We now prove \autoref{thm:read-k-pit}.

\begin{proof}[Proof of \autoref{thm:read-k-pit}]
Consider \autoref{alg:k-read-pit}. We first show that the set $\cH$ it returns hits $A$, and then we bound the size of $\cH$.

By \autoref{cor:read-k-to-roabp}, the polynomial $f$ in the $\vecy_1$ variables is computed by a width-$w^{2k}$ ROABP over the field $\F(\overline{\vecy_1})$. Hence, by \autoref{thm:hitting-set-ROABP}, there exists $\veca_1 \in \cH_1$ such that $f(\veca_1)$ is non-zero over $\F(\overline{\vecy_1})$. Similarly, we can now find $\veca_2 \in \cH_2$ and assign it to the $\vecy_2$ variables and keep the polynomial non-zero, etc.\ all the way up to $\veca_t$. It follows that $(\veca_1, \ldots, \veca_t)$ is an assignment from $\F$ to all the variables such that $f(\veca_1, \ldots, \veca_t)$ is non-zero, as required.

Furthermore, $|\cH_i| = (nw^{2k}d)^{O(\log n)}$, so, $|\cH| = ((nw^{2k}d)^{O(\log n)})^t$. We now bound the number of iterations $t$. Let $T(n)$ denote the number of iterations needed for $n$ variables. We show, by induction on $n$, that
\[
T(n) = c n^{1-1/2^{k-1}},
\]
(for some $c=c(k)$ which we will set in a moment). This will imply the desired bound on $\cH$.

Set $p = 1/2^{k-1}$. After the first iteration, the number of variables we are left with is $n':=n-n^{p}/3^{k^2}$ variables. By the induction hypothesis, we may assume that $T(n') \le c \cdot (n')^{1-p}$. Hence
\[
T(n) \le 1+T(n') \le 1+ c (n')^{1-p} = 1 + c \cdot \left( n - \frac{n^p}{3^{k^2}} \right) ^ {1-p}
\]
and we wish to show that
\[
1 + c \left( n - \frac{n^p}{3^{k^2}} \right) ^ {1-p} \le  c n^{1-p}.
\]
This is equivalent to
\[
\frac{1}{c} \le n^{1-p} - \left( n - \frac{n^p}{3^{k^2}} \right) ^ {1-p},
\]
which is satisfied, by \autoref{lem:iterations-inequality}, as long as
\[
\frac{1-p}{3^{k^2}} \ge \frac{1}{c},
\]
which we can ensure by picking $c=2\cdot 3^{k^2}$ (recall that $p \le 1/2)$.

Finally, since finding the set $\vecy_i$ on each iteration can be done in polynomial time, the running time of the algorithm is dominated by the time required to construct $\cH$, which is $\poly(|\cH|)$.
\end{proof}

\section{Conclusions and Open Problems}\label{sec:conclusions}

In this work, we have obtained the first non-trivial lower bounds and identity testing algorithms for read-$k$ oblivious ABPs. We briefly mention some directions that we find worth pursuing for future research.

The most natural open problem we pose is designing an identity testing algorithm for read-$k$ oblivious ABPs with better running time than the algorithm we presented in this paper. Since for ROABPs (the $k=1$ case) there exist a white-box polynomial time and black-box quasipolynomial-time algorithms, it seems reasonable to hope that the deterioration in the parameters would not be as sharp when $k>1$ (the flip side of this argument, however, is the relative lack of progress in the analogous question in the boolean domain).

Another open problem is obtaining a complete black-box test for read-$k$ oblivious ABPs, in any variable order (that is, without knowing the order in which the variable appear). As we mentioned, for ROABPs there exist a black-box hitting set that works for any variable order \cite{agks15}, whose size is essentially the same as that of the hitting set that was obtained earlier for the known order case \cite{FS13}. In our construction, we need to know the order so that we can pick the per-read-monotone and $k$-regularly-interleaving sequences to which we assign the hitting sets for ROABPs, and simply ``guessing'' those sets would require exponential time. Still, given the progress in obtaining hitting sets in any order for ROABPs, it might be the case that such a construction could follow from our strategy, even using known techniques.

Finally, we turn back to boolean complexity, and ask whether our ideas and techniques can be adapted to attack the problem of constructing pseudorandom generators for read-$k$ oblivious {\em boolean} branching program with sublinear seed length.

\bibliographystyle{customurlbst/alphaurlpp} \bibliography{references}

\newcommand{\etalchar}[1]{$^{#1}$}
\begin{thebibliography}{GMR{\etalchar{+}}12}

\bibitem[AFK85]{AFK85}
Noga Alon, Z.~F{\"u}redi, and M.~Katchalski.
\newblock \href {http://dx.doi.org/10.1016/S0195-6698(85)80028-7} {Separating
  pairs of points by standard boxes}.
\newblock {\em European Journal of Combinatorics}, 6(3):205--210, 1985.

\bibitem[AGKS15]{agks15}
Manindra Agrawal, Rohit Gurjar, Arpita Korwar, and Nitin Saxena.
\newblock \href {http://dx.doi.org/10.1137/140975103} {Hitting-Sets for {ROABP}
  and Sum of Set-Multilinear Circuits}.
\newblock {\em SIAM Journal of Computing}, 44(3):669--697, 2015.
\newblock Pre-print available at \href {http://arxiv.org/abs/1406.7535}
  {\path{arXiv:1406.7535}}.

\bibitem[Agr05]{a05}
Manindra Agrawal.
\newblock \href {http://dx.doi.org/10.1007/11590156_6} {{P}roving {L}ower
  {B}ounds {V}ia {P}seudo-random {G}enerators}.
\newblock In {\em \FSTTCS{2005}}, pages 92--105, 2005.

\bibitem[Ajt05]{Ajtai05}
Mikl{\'{o}}s Ajtai.
\newblock \href {http://dx.doi.org/10.4086/toc.2005.v001a008} {A Non-linear
  Time Lower Bound for Boolean Branching Programs}.
\newblock {\em Theory of Computing}, 1(1):149--176, 2005.
\newblock \pFOCS{1999}.

\bibitem[AR15]{AR15}
V.~Arvind and S.~Raja.
\newblock \href {http://eccc.hpi-web.de/report/2015/176/} {Some Lower Bound
  Results for Set-Multilinear Arithmetic Computations}.
\newblock {\em Electronic Colloquium on Computational Complexity (ECCC)},
  22:176, 2015.

\bibitem[AV08]{av08}
Manindra Agrawal and V.~Vinay.
\newblock \href {http://dx.doi.org/10.1109/FOCS.2008.32} {Arithmetic Circuits:
  A Chasm at Depth Four}.
\newblock In {\em \FOCS{2008}}, pages 67--75, 2008.
\newblock Pre-print available at \parseECCC{TR08/062}.

\bibitem[AvMV11]{amv11}
Matthew Anderson, Dieter van Melkebeek, and Ilya Volkovich.
\newblock {Derandomizing Polynomial Identity Testing for Multilinear
  Constant-Read Formulae}.
\newblock In {\em \CCC{2011}}, pages 273--282, 2011.

\bibitem[AZ04]{proofsfromthebook}
Martin Aigner and G{\"{u}}nter~M. Ziegler.
\newblock \href {http://dx.doi.org/10.1007/978-3-662-44205-0} {{\em Proofs from
  {THE} {BOOK}}}.
\newblock Springer, 2004.

\bibitem[BDVY13]{BDVY13}
Andrej Bogdanov, Zeev Dvir, Elad Verbin, and Amir Yehudayoff.
\newblock \href {http://dx.doi.org/10.4086/toc.2013.v009a007} {Pseudorandomness
  for Width-2 Branching Programs}.
\newblock {\em Theory of Computing}, 9:283--293, 2013.

\bibitem[BHST87]{BHST87}
L{\'{a}}szl{\'{o}} Babai, P{\'{e}}ter Hajnal, Endre Szemer{\'{e}}di, and
  Gy{\"{o}}rgy Tur{\'{a}}n.
\newblock \href {http://dx.doi.org/10.1016/0022-0000(87)90010-9} {A Lower Bound
  for Read-Once-Only Branching Programs}.
\newblock {\em J. Comput. Syst. Sci.}, 35(2):153--162, 1987.

\bibitem[BPW11]{BPW11}
Andrej Bogdanov, Periklis~A. Papakonstantinou, and Andrew Wan.
\newblock \href {http://dx.doi.org/10.1109/FOCS.2011.57} {Pseudorandomness for
  Read-Once Formulas}.
\newblock In {\em \FOCS{2011}}, pages 240--246, 2011.

\bibitem[BRRY14]{BRRY14}
Mark Braverman, Anup Rao, Ran Raz, and Amir Yehudayoff.
\newblock \href {http://dx.doi.org/10.1137/120875673} {Pseudorandom Generators
  for Regular Branching Programs}.
\newblock {\em {SIAM} J. Comput.}, 43(3):973--986, 2014.

\bibitem[BRS93]{BRS93}
Allan Borodin, Alexander~A. Razborov, and Roman Smolensky.
\newblock \href {http://dx.doi.org/10.1007/BF01200404} {On Lower Bounds for
  Read-$k$-Times Branching Programs}.
\newblock {\em Computational Complexity}, 3:1--18, 1993.

\bibitem[BSSV03]{BSSV03}
Paul Beame, Michael~E. Saks, Xiaodong Sun, and Erik Vee.
\newblock \href {http://dx.doi.org/10.1145/636865.636867} {Time-space trade-off
  lower bounds for randomized computation of decision problems}.
\newblock {\em J. {ACM}}, 50(2):154--195, 2003.

\bibitem[De11]{De11}
Anindya De.
\newblock \href {http://dx.doi.org/10.1109/CCC.2011.23} {Pseudorandomness for
  Permutation and Regular Branching Programs}.
\newblock In {\em \CCC{2011}}, pages 221--231, 2011.

\bibitem[DL78]{DL78}
Richard~A. DeMillo and Richard~J. Lipton.
\newblock \href {http://dx.doi.org/10.1016/0020-0190(78)90067-4} {{A
  Probabilistic Remark on Algebraic Program Testing}}.
\newblock {\em Information Processing Letters}, 7(4):193--195, 1978.

\bibitem[DS07]{DS07}
Zeev Dvir and Amir Shpilka.
\newblock \href {http://dx.doi.org/10.1137/05063605X} {Locally Decodable Codes
  with Two Queries and Polynomial Identity Testing for Depth 3 Circuits}.
\newblock {\em {SIAM} J. Comput.}, 36(5):1404--1434, 2007.
\newblock \pSTOC{2005}.

\bibitem[DSY09]{DSY09}
Zeev Dvir, Amir Shpilka, and Amir Yehudayoff.
\newblock \href {http://dx.doi.org/10.1137/080735850} {Hardness-Randomness
  Tradeoffs for Bounded Depth Arithmetic Circuits}.
\newblock {\em {SIAM} J. Comput.}, 39(4):1279--1293, 2009.

\bibitem[ES35]{ES35}
Paul Erd\H{o}s and George Szekeres.
\newblock \href {http://dx.doi.org/10.1007/978-0-8176-4842-8_3} {A
  combinatorial problem in geometry}.
\newblock {\em Compositio Math.}, 2:463--470, 1935.

\bibitem[FLMS14]{FLMS13}
Herv{\'{e}} Fournier, Nutan Limaye, Guillaume Malod, and Srikanth Srinivasan.
\newblock \href {http://doi.acm.org/10.1145/2591796.2591824} {Lower bounds for
  depth 4 formulas computing iterated matrix multiplication}.
\newblock In {\em \STOC{2014}}, pages 128--135, 2014.
\newblock Pre-print available at \parseECCC{TR13/100}.

\bibitem[For14]{forbesphdthesis}
Michael Forbes.
\newblock \href {http://hdl.handle.net/1721.1/89843} {{\em {P}olynomial
  {I}dentity {T}esting of {R}ead-{O}nce {O}blivious {A}lgebraic {B}ranching
  {P}rograms}}.
\newblock PhD thesis, Massachusetts Institute of Technology, 2014.

\bibitem[FS13a]{FS13b}
Michael~A. Forbes and Amir Shpilka.
\newblock \href {http://dx.doi.org/10.1007/978-3-642-40328-6_37} {Explicit
  Noether Normalization for Simultaneous Conjugation via Polynomial Identity
  Testing}.
\newblock In {\em \RANDOM{2013}}, volume 8096 of {\em Lecture Notes in Computer
  Science}, pages 527--542. Springer, 2013.

\bibitem[FS13b]{FS13}
Michael~A. Forbes and Amir Shpilka.
\newblock \href {http://dx.doi.org/10.1109/FOCS.2013.34} {Quasipolynomial-Time
  Identity Testing of Non-commutative and Read-Once Oblivious Algebraic
  Branching Programs}.
\newblock In {\em \FOCS{2013}}, pages 243--252, 2013.
\newblock \farXiv{1209.2408}.

\bibitem[FSS14]{FSS14}
Michael~A. Forbes, Ramprasad Saptharishi, and Amir Shpilka.
\newblock \href {http://dx.doi.org/10.1145/2591796.2591816} {Hitting sets for
  multilinear read-once algebraic branching programs, in any order}.
\newblock In {\em \STOC{2014}}, pages 867--875, 2014.

\bibitem[GKKS13]{gkks13b}
Ankit Gupta, Pritish Kamath, Neeraj Kayal, and Ramprasad Saptharishi.
\newblock \href {http://dx.doi.org/10.1109/FOCS.2013.68} {{Arithmetic Circuits:
  {A} Chasm at Depth Three}}.
\newblock In {\em \FOCS{2013}}, pages 578--587, 2013.
\newblock Pre-print available at \parseECCC{TR13/026}.

\bibitem[GKKS14]{gkks13}
Ankit Gupta, Pritish Kamath, Neeraj Kayal, and Ramprasad Saptharishi.
\newblock \href {http://dx.doi.org/10.1145/2629541} {Approaching the Chasm at
  Depth Four}.
\newblock {\em Journal of the ACM}, 61(6):33:1--33:16, 2014.
\newblock \pCCC{2013}.
\newblock Pre-print available at \parseECCC{TR12/098}.

\bibitem[GKST15]{GKST15}
Rohit Gurjar, Arpita Korwar, Nitin Saxena, and Thomas Thierauf.
\newblock \href {http://dx.doi.org/10.4230/LIPIcs.CCC.2015.323} {Deterministic
  Identity Testing for Sum of Read-once Oblivious Arithmetic Branching
  Programs}.
\newblock In {\em \CCC{2015}}, pages 323--346, 2015.
\newblock Pre-print available at \href {http://arxiv.org/abs/1411.7341}
  {\path{arXiv:1411.7341}}.

\bibitem[GMR{\etalchar{+}}12]{GMRTV12}
Parikshit Gopalan, Raghu Meka, Omer Reingold, Luca Trevisan, and Salil~P.
  Vadhan.
\newblock \href {http://dx.doi.org/10.1109/FOCS.2012.77} {Better Pseudorandom
  Generators from Milder Pseudorandom Restrictions}.
\newblock In {\em \FOCS{2012}}, pages 120--129. {IEEE} Computer Society, 2012.

\bibitem[HS80]{HS80}
Joos Heintz and Claus-Peter Schnorr.
\newblock \href {http://dx.doi.org/10.1145/800141.804674} {{Testing Polynomials
  which Are Easy to Compute (Extended Abstract)}}.
\newblock In {\em \STOC{1980}}, pages 262--272, 1980.

\bibitem[IMZ12]{IMZ12}
Russell Impagliazzo, Raghu Meka, and David Zuckerman.
\newblock \href {http://dx.doi.org/10.1109/FOCS.2012.78} {Pseudorandomness from
  Shrinkage}.
\newblock In {\em \FOCS{2012}}, pages 111--119, 2012.

\bibitem[INW94]{INW94}
Russell Impagliazzo, Noam Nisan, and Avi Wigderson.
\newblock \href {http://dx.doi.org/10.1145/195058.195190} {Pseudorandomness for
  network algorithms}.
\newblock In {\em \STOC{1994}}, pages 356--364, 1994.

\bibitem[JQS10]{JQS10}
Maurice~J. Jansen, Youming Qiao, and Jayalal Sarma.
\newblock \href {http://eccc.hpi-web.de/report/2010/084} {Deterministic
  Identity Testing of Read-Once Algebraic Branching Programs}.
\newblock {\em Electronic Colloquium on Computational Complexity {(ECCC)}},
  17:84, 2010.

\bibitem[KI04]{ki03}
Valentine Kabanets and Russell Impagliazzo.
\newblock \href {http://dx.doi.org/10.1007/s00037-004-0182-6} {{D}erandomizing
  Polynomial Identity Tests Means Proving Circuit Lower Bounds}.
\newblock {\em Computational Complexity}, 13(1-2):1--46, 2004.
\newblock \pSTOC{2003}.

\bibitem[KLSS14]{KLSS}
Neeraj Kayal, Nutan Limaye, Chandan Saha, and Srikanth Srinivasan.
\newblock \href {http://dx.doi.org/10.1109/FOCS.2014.15} {{An Exponential Lower
  Bound for Homogeneous Depth Four Arithmetic Circuits}}.
\newblock In {\em \FOCS{2014}}, 2014.
\newblock Pre-print available at \parseECCC{TR14/005}.

\bibitem[KMSV13]{kmsv13}
Zohar~Shay Karnin, Partha Mukhopadhyay, Amir Shpilka, and Ilya Volkovich.
\newblock \href {http://dx.doi.org/10.1137/110824516} {Deterministic Identity
  Testing of Depth-4 Multilinear Circuits with Bounded Top Fan-in}.
\newblock {\em SIAM Journal of Computing}, 42(6):2114--2131, 2013.
\newblock \pSTOC{2010}.
\newblock Pre-print available at \parseECCC{TR09-116}.

\bibitem[KNP11]{KNP11}
Michal Kouck{\'{y}}, Prajakta Nimbhorkar, and Pavel Pudl{\'{a}}k.
\newblock \href {http://dx.doi.org/10.1145/1993636.1993672} {Pseudorandom
  generators for group products: extended abstract}.
\newblock In {\em \STOC{2011}}, pages 263--272, 2011.

\bibitem[KNS15]{KNS15}
Neeraj Kayal, Vineet Nair, and Chandan Saha.
\newblock \href {http://eccc.hpi-web.de/report/2015/154/} {Separation between
  Read-once Oblivious Algebraic Branching Programs (ROABPs) and Multilinear
  Depth Three Circuits}.
\newblock {\em Electronic Colloquium on Computational Complexity (ECCC)},
  22:154, 2015.

\bibitem[Koi12]{koiran}
Pascal Koiran.
\newblock \href {http://dx.doi.org/10.1016/j.tcs.2012.03.041} {Arithmetic
  Circuits: The Chasm at Depth Four Gets Wider}.
\newblock {\em Theoretical Computer Science}, 448:56--65, 2012.
\newblock Pre-print available at \href {http://arxiv.org/abs/1006.4700}
  {\path{arXiv:1006.4700}}.

\bibitem[Kru53]{Kru53}
Joseph~B Kruskal.
\newblock \href {http://dx.doi.org/10.1090/S0002-9939-1953-0053256-2}
  {Monotonic subsequences}.
\newblock {\em Proceedings of the American Mathematical Society},
  4(2):264--274, 1953.

\bibitem[KS07]{KS07}
Neeraj Kayal and Nitin Saxena.
\newblock \href {http://dx.doi.org/10.1007/s00037-007-0226-9} {Polynomial
  Identity Testing for Depth 3 Circuits}.
\newblock {\em Computational Complexity}, 16(2):115--138, 2007.

\bibitem[KS09]{ks09}
Neeraj Kayal and Shubhangi Saraf.
\newblock \href {http://dx.doi.org/10.1109/FOCS.2009.67} {{Blackbox polynomial
  identity testing for depth-$3$ circuits}}.
\newblock In {\em \FOCS{2009}}, 2009.

\bibitem[KS14a]{KS14a}
Mrinal Kumar and Shubhangi Saraf.
\newblock \href {http://doi.acm.org/10.1145/2591796.2591827} {The limits of
  depth reduction for arithmetic formulas: it's all about the top fan-in}.
\newblock In {\em \STOC{2014}}, pages 136--145, 2014.
\newblock Pre-print available at \parseECCC{TR13/068}.

\bibitem[KS14b]{KS14}
Mrinal Kumar and Shubhangi Saraf.
\newblock \href {http://dx.doi.org/10.1109/FOCS.2014.46} {{On the power of
  homogeneous depth $4$ arithmetic circuits}}.
\newblock In {\em \FOCS{2014}}, 2014.
\newblock Pre-print available at \parseECCC{TR14/045}.

\bibitem[KSS14]{KSS13}
Neeraj Kayal, Chandan Saha, and Ramprasad Saptharishi.
\newblock \href {http://doi.acm.org/10.1145/2591796.2591847} {A
  super-polynomial lower bound for regular arithmetic formulas}.
\newblock In {\em \STOC{2014}}, pages 146--153, 2014.
\newblock Pre-print available at \parseECCC{TR13/091}.

\bibitem[KSS15]{KSS15}
Swastik Kopparty, Shubhangi Saraf, and Amir Shpilka.
\newblock \href {http://dx.doi.org/10.1007/s00037-015-0102-y} {Equivalence of
  Polynomial Identity Testing and Polynomial Factorization}.
\newblock {\em Computational Complexity}, 24(2):295--331, 2015.
\newblock \pCCC{2014}.

\bibitem[KUW86]{KUW86}
Richard~M. Karp, Eli Upfal, and Avi Wigderson.
\newblock \href {http://dx.doi.org/10.1007/BF02579407} {Constructing a perfect
  matching is in random {NC}}.
\newblock {\em Combinatorica}, 6(1):35--48, 1986.
\newblock \pSTOC{1985}.

\bibitem[Mul12]{Mulmuley12}
Ketan Mulmuley.
\newblock \href {http://dx.doi.org/10.1109/FOCS.2012.15} {Geometric Complexity
  Theory {V:} Equivalence between Blackbox Derandomization of Polynomial
  Identity Testing and Derandomization of Noether's Normalization Lemma}.
\newblock In {\em \FOCS{2012}}, pages 629--638, 2012.

\bibitem[MVV87]{MVV87}
Ketan Mulmuley, Umesh~V. Vazirani, and Vijay~V. Vazirani.
\newblock \href {http://dx.doi.org/10.1007/BF02579206} {Matching is as easy as
  matrix inversion}.
\newblock {\em Combinatorica}, 7(1):105--113, 1987.
\newblock \pSTOC{1987}.

\bibitem[Nis91]{nis91}
Noam Nisan.
\newblock \href {http://dx.doi.org/10.1145/103418.103462} {{Lower bounds for
  non-commutative computation}}.
\newblock In {\em \STOC{1991}}, pages 410--418, 1991.
\newblock Available on
  \href{http://citeseerx.ist.psu.edu/viewdoc/summary?doi=10.1.1.17.5067}{\tt
  citeseer:10.1.1.17.5067}.

\bibitem[Nis92]{Nisan92}
Noam Nisan.
\newblock \href {http://dx.doi.org/10.1007/BF01305237} {Pseudorandom generators
  for space-bounded computation}.
\newblock {\em Combinatorica}, 12(4):449--461, 1992.

\bibitem[Oko91]{Oko91}
E.A. Okolnishnikova.
\newblock \href
  {http://www.thi.informatik.uni-frankfurt.de/~jukna/boolean/Russians/Okolnishnikova-1991.pdf}
  {Lower bounds on the complexity of realization of characteristic functions of
  binary codes by branching programs}.
\newblock {\em Metody Diskretnogo Analiza}, 51:61--83, 1991.

\bibitem[OSV15]{osv15}
Rafael Oliveira, Amir Shpilka, and Ben~Lee Volk.
\newblock \href {http://dx.doi.org/10.4230/LIPIcs.CCC.2015.304} {Subexponential
  Size Hitting Sets for Bounded Depth Multilinear Formulas}.
\newblock In {\em \CCC{2015}}, pages 304--322, 2015.
\newblock Pre-print available at \href {http://arxiv.org/abs/1411.7492}
  {\path{arXiv:1411.7492}}.

\bibitem[RS05]{RS05}
Ran Raz and Amir Shpilka.
\newblock \href {http://dx.doi.org/10.1007/s00037-005-0188-8} {Deterministic
  polynomial identity testing in non-commutative models}.
\newblock {\em Computational Complexity}, 14(1):1--19, 2005.
\newblock \pCCC{2004}.

\bibitem[RSV13]{RSV13}
Omer Reingold, Thomas Steinke, and Salil~P. Vadhan.
\newblock \href {http://dx.doi.org/10.1007/978-3-642-40328-6_45}
  {Pseudorandomness for Regular Branching Programs via Fourier Analysis}.
\newblock In {\em \RANDOM{2013}}, pages 655--670, 2013.

\bibitem[RY09]{raz-yehudayoff}
Ran Raz and Amir Yehudayoff.
\newblock \href {http://dx.doi.org/10.1007/s00037-009-0270-8} {Lower Bounds and
  Separations for Constant Depth Multilinear Circuits}.
\newblock {\em Computational Complexity}, 18(2):171--207, 2009.
\newblock \pCCC{2008}.
\newblock Pre-print available at \parseECCC{TR08/006}.

\bibitem[Sch80]{S80}
Jacob~T. Schwartz.
\newblock \href {http://dx.doi.org/10.1145/322217.322225} {{F}ast
  {P}robabilistic {A}lgorithms for {V}erification of {P}olynomial
  {I}dentities}.
\newblock {\em Journal of the ACM}, 27(4):701--717, 1980.

\bibitem[SS12]{SS12}
Nitin Saxena and C.~Seshadhri.
\newblock \href {http://dx.doi.org/10.1137/10848232} {Blackbox Identity Testing
  for Bounded Top-Fanin Depth-3 Circuits: The Field Doesn't Matter}.
\newblock {\em {SIAM} J. Comput.}, 41(5):1285--1298, 2012.
\newblock \pSTOC{2011}.

\bibitem[Ste12]{Steinke12}
Thomas Steinke.
\newblock \href {http://eccc.hpi-web.de/report/2012/083} {Pseudorandomness for
  Permutation Branching Programs Without the Group Theory}.
\newblock {\em Electronic Colloquium on Computational Complexity {(ECCC)}},
  19:83, 2012.

\bibitem[SV10]{sv-icalp10}
Amir Shpilka and Ilya Volkovich.
\newblock \href {http://dx.doi.org/10.1007/978-3-642-14165-2_35} {On the
  Relation between Polynomial Identity Testing and Finding Variable Disjoint
  Factors}.
\newblock In {\em \ICALP{2010}}, pages 408--419, 2010.
\newblock Pre-print available at \parseECCC{TR10-036}.

\bibitem[SV11]{sv11}
Shubhangi Saraf and Ilya Volkovich.
\newblock \href {http://dx.doi.org/10.1145/1993636.1993693} {{Black-box
  identity testing of depth-4 multilinear circuits}}.
\newblock In {\em \STOC{2011}}, pages 421--430, 2011.

\bibitem[SV15]{SV15}
Amir Shpilka and Ilya Volkovich.
\newblock \href {http://dx.doi.org/10.1007/s00037-015-0105-8} {Read-once
  polynomial identity testing}.
\newblock {\em Computational Complexity}, 24(3):477--532, 2015.
\newblock \pSTOC{2008}.

\bibitem[SVW14]{SVW14}
Thomas Steinke, Salil~P. Vadhan, and Andrew Wan.
\newblock \href {http://dx.doi.org/10.4230/LIPIcs.APPROX-RANDOM.2014.885}
  {Pseudorandomness and Fourier Growth Bounds for Width-3 Branching Programs}.
\newblock In {\em \RANDOM{2014}}, pages 885--899, 2014.

\bibitem[SY10]{sy}
Amir Shpilka and Amir Yehudayoff.
\newblock \href {http://dx.doi.org/http://dx.doi.org/10.1561/0400000039}
  {Arithmetic Circuits: A survey of recent results and open questions}.
\newblock {\em Foundations and Trends in Theoretical Computer Science},
  5:207--388, March 2010.

\bibitem[Tav15]{Tav13}
S{\'{e}}bastien Tavenas.
\newblock \href {http://dx.doi.org/10.1016/j.ic.2014.09.004} {Improved bounds
  for reduction to depth 4 and depth 3}.
\newblock {\em Inf. Comput.}, 240:2--11, 2015.
\newblock \pMFCS{2013}.

\bibitem[Tha98]{Thathachar98}
Jayram~S. Thathachar.
\newblock \href {http://dx.doi.org/10.1145/276698.276881} {On Separating the
  Read-k-Times Branching Program Hierarchy}.
\newblock In {\em \STOC{1998}}, pages 653--662, 1998.

\bibitem[VSBR83]{vsbr83}
Leslie~G. Valiant, Sven Skyum, S.~Berkowitz, and Charles Rackoff.
\newblock \href {http://dx.doi.org/10.1137/0212043} {{Fast Parallel Computation
  of Polynomials Using Few Processors}}.
\newblock {\em SIAM Journal of Computing}, 12(4):641--644, 1983.
\newblock \pMFCS{1981}.

\bibitem[Weg88]{Wegener88}
Ingo Wegener.
\newblock \href {http://dx.doi.org/10.1145/42282.46161} {On the complexity of
  branching programs and decision trees for clique functions}.
\newblock {\em J. {ACM}}, 35(2):461--471, 1988.

\bibitem[Z{\'{a}}k84]{Zak84}
Stanislav Z{\'{a}}k.
\newblock \href {http://dx.doi.org/10.1007/BFb0030340} {An Exponential Lower
  Bound for One-Time-Only Branching Programs}.
\newblock In {\em \MFCS{1984}}, volume 176 of {\em Lecture Notes in Computer
  Science}, pages 562--566. Springer, 1984.

\bibitem[Zip79]{Z79}
Richard Zippel.
\newblock \href {http://dx.doi.org/10.1007/3-540-09519-5_73} {Probabilistic
  algorithms for sparse polynomials}.
\newblock In {\em Symbolic and Algebraic Computation, {EUROSAM} '79, An
  International Symposiumon Symbolic and Algebraic Computation}, volume~72 of
  {\em Lecture Notes in Computer Science}, pages 216--226. Springer, 1979.

\end{thebibliography}

\appendix

\section{Proof of Lemma \ref{lem:iterations-inequality}}
\label{app:iterations-inequality}

For convenience, let us first recall the statement of the lemma.

\begin{lemma}[\autoref{lem:iterations-inequality}, restated]
Let $p$ be a real number such that $0<p<1$ and $r$ be a positive integer. For any $n \in \N$,
\[
n^{1-p} - (n-n^p/r)^{1-p} \ge (1-p)/r.
\]
\end{lemma}

\begin{proof}
Define $f:\R^+ \to \R^+$ by $f(x) = x^{1-p} - (x-x^p/r)^{1-p}$. We show that this real function is non-increasing for non-negative $x$, and that its limit as $x$ tends to infinity is $(1-p)/r$, which implies the statement of the lemma.

To show that $f$ is non-increasing, we will show that its derivative is non-positive. Note that 
\[
f'(x) = (1-p)x^{-p} - (1-p) \inparen{1-\frac{px^{p-1}}{r}} \inparen{x-\frac{x^p}{r}}^{-p}.
\]
To show that $f'(x) \le 0$ for all $x$, it thus suffices, after some rearrangements, to prove the inequality
\begin{equation}
\label{eq:derivative-negative}
\inparen{x-\frac{x^p}{r}}^{p} \le x^{p} \inparen{1-\frac{px^{p-1}}{r}}.
\end{equation}
We have that
\[
\inparen{x-\frac{x^p}{r}}^{p} = x^p \inparen{1-\frac{x^{p-1}}{r}}^p,
\]
and thus after dividing by $x^p$, \eqref{eq:derivative-negative} follows as a corollary of the well-known inequality $(1-y)^s \le 1-sy$ for $y>0$ and $0 < s < 1$, that can be proved using the Taylor expansion of $(1-y)^s$ around $0$.

In order to calculate the limit, observe that,
\[
f(x) = x^{1-p} \cdot \inparen{ 1-\inparen{1-\frac{x^{p-1}}{r}}^{1-p} } =  \frac{1-\inparen{1-\frac{x^{p-1}}{r} }^{1-p}} {x^{p-1}},
\]
so by L'H\^{o}pital's Rule we get that
\[
\lim_{x \to \infty} \frac{1-\inparen{1-\frac{x^{p-1}}{r} }^{1-p}} {x^{p-1}}=
\lim_{x \to \infty} \frac{\frac{-(p-1)^2}{r} \cdot x^{p-2} \cdot \inparen{1-\frac{x^{p-1}}{r}}^{-p}} {(p-1)x^{p-2}} = \frac{1-p}{r}. \qedhere
\]
\end{proof}

\end{document}